\documentclass[screen,sigconf]{acmart}

\makeatletter
\def\@myparfont{\bfseries}
\newcommand\myparagraph{\def\@toclevel{4}%
  \@startsection{paragraph}{4}{\z@}%
  {-.2\baselineskip \@plus -2\p@ \@minus -.2\p@}%
  {-3.5\p@}%
  {\ACM@NRadjust{\@myparfont}}}
\makeatother
\newcommand{\bigO}[1]{\mathchoice{O\left(#1\right)}{O(#1)}{O(#1)}{O(#1)}} % big O for complexity
\newcommand{\softO}[1]{\mathchoice{\tilde{O}\left(#1\right)}{\tilde O(#1)}{O\tilde{~}(#1)}{O\tilde{~}(#1)}} % soft O for complexity
\newcommand{\expmm}{\omega} % exponent for the cost of matrix multiplication
\newcommand{\ZZ}{\mathbb{Z}} % relative integers
\newcommand{\NN}{\mathbb{N}} %  integers
\newcommand{\field}{\mathbb{K}} % base field
\newcommand{\ring}{\mathbb{A}} % a ring
\newcommand{\otherring}{\mathbb{B}} % another ring
\newcommand{\matspace}[2]{\field^{#1 \times #2}} % scalar matrix space
\newcommand{\genBy}[1]{\langle #1 \rangle} % ideal/module generated by #1
\newcommand{\Span}{\operatorname{Span}} % vector space span
\newcommand{\xring}{\field[x]} % univariate polynomial ring in x
\newcommand{\xmatspace}[2]{\xring^{#1 \times #2}} % univariate polynomial matrix space in x
\newcommand{\yring}{\field[y]} % univariate polynomial ring in y
\newcommand{\ymatspace}[2]{\yring^{#1 \times #2}} % univariate polynomial matrix space in y
\newcommand{\xyring}{\field[x,y]} % bivariate polynomial ring in x,y
\newcommand{\xynotring}{\field[x,y_0]} % bivariate polynomial ring in x,y_0
\newcommand{\trsp}[1]{#1^\mathsf{T}} % transpose
\newcommand{\diag}[1]{\operatorname{diag}(#1)}  % diagonal matrix with diagonal entries #1
\newcommand{\ident}[1]{\mathrm{I}_{#1}} % identity matrix of size #1 x #1
\newcommand{\sddots}{\raisebox{3pt}{$\scalebox{.75}{$\ddots$}$}} % to scale ddots, e.g. in smallmatrix
\newcommand{\svdots}{\raisebox{3pt}{$\scalebox{.75}{$\vdots$}$}} % to scale vdots, e.g. in smallmatrix
\makeatletter
\newcommand*{\rem}{%
  \nonscript\mskip-\medmuskip\mkern5mu%
  \mathbin{\operator@font rem}\penalty900\mkern5mu%
  \nonscript\mskip-\medmuskip
}
\makeatother
\newcommand{\module}{\mathcal{M}} % typically K[x]-module of relations
\newcommand{\nodule}{\mathcal{N}} % typically K[y]-module of relations
\newcommand{\ideal}{\mathcal{I}} % typically ideal <f(x), y-a(x)>
\newcommand{\idealGens}{\genBy{f(x),y-a(x)}} % the typical ideal
\newcommand{\ddelta}{d}
\newcommand{\mmu}{m}
\newcommand{\dd}{\delta}
\newcommand{\mm}{\mu}
\newcommand{\witnessM}{\Delta} % big polynomials for Zariski open set
\newcommand{\witnessN}{\Phi} % big polynomials for Zariski open set
\usepackage{bm} % for \bm{..}
\usepackage{tikz} %tikzpicture
\usepackage[shortlabels]{enumitem}

\theoremstyle{acmplain}
\newtheorem*{theorem*}{Theorem}
\newtheorem{theorem}{Theorem}[section]
\theoremstyle{acmdefinition}

\newtheorem{definition}[theorem]{Definition}

\usepackage[capitalize]{cleveref}
\crefname{enumi}{Step}{Steps}
\crefname{theorem}{Theorem}{Theorems}
\AddToHook{cmd/lemma/before}{\crefalias{theorem}{lemma}}
\crefname{lemma}{Lemma}{Lemmas}
\AddToHook{cmd/definition/before}{\crefalias{theorem}{definition}}
\crefname{definition}{Definition}{Definitions}
\AddToHook{cmd/proposition/before}{\crefalias{theorem}{proposition}}
\crefname{proposition}{Proposition}{Propositions}
\AddToHook{cmd/corollary/before}{\crefalias{theorem}{corollary}}
\crefname{corollary}{Corollary}{Corollaries}
\AddToHook{cmd/appendix/before}{\crefalias{section}{appendix}}
\crefname{section}{Section}{Sections}
\begin{document}

%%%%%%%%%%%%%%%%%%%%%%%%%%%%%%%%%%%%%%%%%%%%%%%%%%%%%%%%%%%%%%%%%%%%%%%%%%%%%%%%%%%%%%%%%%%%%%
%%%%%%%%%%%%%%%%%%%%%%%%%%%%%%%%%%%%%%%%%%%%%%%%%%%%%%%%%%%%%%%%%%%%%%%%%%%%%%%%%%%%%%%%%%%%%%
%
%       INFOS
%
%%%%%%%%%%%%%%%%%%%%%%%%%%%%%%%%%%%%%%%%%%%%%%%%%%%%%%%%%%%%%%%%%%%%%%%%%%%%%%%%%%%%%%%%%%%%%%

\title{Faster modular composition using two relation matrices}

\author{Vincent Neiger}
\orcid{0000-0002-8311-9490}
\affiliation{
\institution{Sorbonne Université, CNRS, LIP6}
\city{F-75005 Paris}
  \country{France}
}

\author{Bruno Salvy}
\orcid{0000-0002-4313-0679}
\affiliation{
  % \institution{Inria, ENS de Lyon, CNRS,}
  % \department{UCBL, LIP UMR 5668}
   \institution{Inria, ENS de Lyon, CNRS, UCBL, LIP UMR 5668}
  %\streetaddress{46, Allée d'Italie}
  \city{Lyon}
  \country{France}
}

\author{\'Eric Schost}
\orcid{0000-0002-8638-9357}
\affiliation{
  \institution{University of Waterloo}
  \department{Cheriton School of Computer Science}
  %\institution{University of Waterloo,  School of Computer Science} 
    \city{Waterloo}
  \country{Canada}
}

\author{Gilles Villard}
\orcid{0000-0003-1936-7155}
\affiliation{
  % \institution{CNRS, ENS de Lyon, Inria,}
  % \department{UCBL, LIP UMR 5668}
  \institution{CNRS, ENS de Lyon, Inria, UCBL, LIP UMR 5668}
  %\streetaddress{46, Allée d'Italie}
  \city{Lyon}
  \country{France}
}

\begin{abstract}
Modular composition is the problem of computing the composition of two
univariate polynomials modulo a third one. For a long time, the
algebraic algorithm with the best complexity bound for this problem was that of Brent and Kung (1978).
Recently, we improved this algorithm by computing and using a
polynomial matrix that encodes a certain basis of algebraic relations between
the polynomials. This is further improved here by making use of two polynomial
matrices of smaller dimension. Under genericity assumptions on the input, this
results in an algorithm using $\tilde{O}(n^{(\omega+3)/4})$ arithmetic
operations in the base field, where $\omega$ is the exponent of matrix
multiplication. With naive matrix multiplication, this is $\tilde{O}(n^{3/2})$,
while with the best currently known exponent $\omega$ this is $O(n^{1.343})$,
improving upon the previously most efficient algorithms.
\end{abstract}

%% acmart stuff <<:
\copyrightyear{2026}
\acmYear{2026}
\setcopyright{cc}
\setcctype{by}
\acmConference[ISSAC '26]{51st International Symposium on Symbolic and Algebraic Computation}{July 13--17, 2026}{Oldenburg, Germany}
\acmBooktitle{51st International Symposium on Symbolic and Algebraic Computation (ISSAC '26), July 13--17, 2026, Oldenburg, Germany}
\acmDOI{10.1145/3815436.3815451}
\acmISBN{979-8-4007-2595-1/2026/07}


% :>> acmart stuff

\maketitle

%%%%%%%%%%%%%%%%%%%%%%%%%%%%%%%%%%%%%%%%%%%%%%%%%%%%%%%%%%%%%%%%%%%%%%%%%%%%%%%%%%%
%
%              INTRODUCTION
%
%%%%%%%%%%%%%%%%%%%%%%%%%%%%%%%%%%%%%%%%%%%%%%%%%%%%%%%%%%%%%%%%%%%%%%%%%%%%%%%%%%%

\section{Introduction}
\label{sec:intro}

Given three polynomials $a,f\in\field[x]$ and $g\in\field[y]$ with
coefficients in a field $\field$, modular
composition asks to compute $g(a)\bmod f$.

Quasi-optimal solutions to this problem are known in special cases and
complexity models. In particular, a famous algorithm of Kedlaya and
Umans solves it in $n^{1+\epsilon}\log(q)^{1+o(1)}$ bit
operations when $\field=\mathbb F_q$ is a finite field~\cite{KU11},
where \(n\) is the degree of \(f\);
see also \cite{HoeLec20c,BGKM23} for detailed analyses and further improvements.
Very recently, the case of formal power series (where
$f=x^n$) was solved in quasi-linear complexity by Kinoshita and
Li~\cite{KinoshitaLi2024}, in the algebraic model of computation,
where the complexity is expressed in terms of the number of arithmetic operations
in~$\field$. 

In the general case and in the algebraic complexity model, the first
non-trivial algorithm was due to Brent and Kung~\cite{BK78}. It has
been an open problem for a long time to improve upon it and reach a
complexity exponent below~$3/2$
% \cite[Open pb~2.4]{BurgisserClausenShokrollahi1997};
% \cite[Research pb~12.19]{GaGe99}.
\cite{BurgisserClausenShokrollahi1997,GaGe99}.
This was achieved
recently (see~\cite{NSSV24} where more motivation can be found). The
present article brings further
improvement
to this general situation.

The main parameter in complexity estimates is $n=\deg(f)$. Indeed, we may focus
on the case $\deg(a)<n$, to which the general case reduces via $g(a) \bmod f =
g(a\bmod f) \bmod f$; and up to using a $y^n$-adic expansion of~$g$, we may
restrict to \(\deg(g) < n\). We use the $\softO{\cdot}$ notation, which hides
factors that are polylogarithmic in \(n\).

Recall that using fast polynomial multiplication, sums, products and
divisions
modulo~$f$ can be computed in~$\softO{n}$ operations 
in~\(\field\)~\cite{GaGe99}.
Thus, the most natural algorithm for
modular composition, based on Horner evaluation, has
complexity~$\softO{n^2}$, which is in agreement with the number
$\Theta(n^2)$ of coefficients it computes.

\myparagraph{Brent and Kung's algorithm.}

Hereafter, we use standard notation such as $\xring_{<m}$ for the set of
univariate polynomials in $x$ with coefficients in $\field$ and degree less
than $m$, and $\xyring_{< (m,d)}$ for the set of bivariate polynomials in $x,y$
of bidegree less than $(m,d)$. 

Brent and Kung~\cite{BK78} introduced the following \emph{baby steps-giant
steps} algorithm, where $m=\lceil\sqrt n\rceil$. 
\begin{enumerate}[({1.}1),leftmargin=0.2cm,itemindent=0.7cm]
  \item \label{BK:baby-steps} Reduce the polynomials
    $a,a^{2},\ldots,a^{m}$ modulo $f$;
  \item Build \(\bar{G}\) in \(\field[y_0,y_1]_{<(m,m)}\) such that
    \(g = \bar{G}(y,y^m)\), and write \(\bar{G} = \bar{g}_0(y_0) +
    \bar{g}_1(y_0) y_1 + \cdots + \bar{g}_{m-1}(y_0) y_1^{m-1}\);
  \item Using the $a^i$'s of Step~\ref{BK:baby-steps}, 
    compute \(\bar{G}(a, y_1) \bmod f\) via the
    multiplication of two matrices in $\field^{m\times m}$ and $\field^{m\times n}$,
    which simultaneously yields all $\bar{g}_i(a)$ modulo $f$;
  \item Deduce \(g(a) = \bar{G}(a, a^m) \bmod f\) using Horner evaluation by evaluating $y_1$ at
    $a^m$ modulo~$f$.
\end{enumerate}

This method uses~$\softO{m^{\expmm+1}} =
\softO{n^{(\expmm+1)/2}}$ operations in $\field$, later
improved by Huang and Pan using fast rectangular matrix
multiplication~\cite{HuPa98}. The total number of coefficients
computed by this method is of order~$n^{3/2}$, which explains the 
previous barrier of~$3/2$ on the exponent of the algebraic complexity.

\myparagraph{Polynomial matrix algorithms} have recently been
introduced in this context, to make use of bivariate polynomials.
A simplified version of the
algorithm in~ \cite{NSSV24}, which applies when $f(0)\neq0$ and the
input \(a\) is generic (in the Zariski sense), can be sketched as
follows. It uses a smaller parameter $m = \lceil n^ {1/3}\rceil$, and
 the complementary parameter $d=\lceil n/m\rceil$ which is of
order \(\Theta(m^2) = \Theta(n^{2/3})\).

In what follows, for a polynomial $g=\sum_{i} g_i x^i$ in $\xring$ and 
nonnegative integers $u,k$, we write $[g]_u^k=
g_u+g_{u+1}x+\cdots+g_{u+k}x^{k}$; $g\rem f$ denotes the
remainder of $g$ in the Euclidean division
by $f$.

\begin{enumerate}[({2.}1),leftmargin=0.2cm,itemindent=0.7cm]
  \item\label{old-algo:baby-steps}
  Reduce 
  $a^2,\ldots,a^m, a^{2m},\ldots,a^{(m-1)m}$
  modulo $f$;
  \item \label{old-algo:step2}
    Compute $[x^ia^{-k}\rem f]_0^{m-1}$ for $0 \le i < m$ and $0 \le k
    < 2d$;
  \item \label{old-algo:basis} Deduce a basis of \(y\)-degree at most \(d\) for the $\field[y]$-module 
    \[\mathcal M_m = \{p(x,y)\in\xyring_{<(m,\cdot)} \mid p(x,a)=0\bmod f\};\]
  \item \label{old-algo:step4} Reduce~$g$ by this basis to obtain
  $G\in\xyring_{<(m,d)}$  such that $G(x,a)=g(a)$ modulo $f$;
  \item \label{old-algo:step4.5}
    Build \(\bar{G}\) in \(\field[x,y_0,y_1]_{<(m,m,m)}\) such that \(G = \bar{G}(x,y,y^m)\),
    and write $\bar{G} = \sum_{0\le i < m} \bar{g}_i(x,y_0) y_1^i$;
  \item \label{old-algo:step5} Using the $a^i$'s of Step~\ref{old-algo:baby-steps}, 
    compute \(\bar{G}(x, a, y_1) \bmod f\), via the multiplication of
    matrices in $\xmatspace{m}{m}_{<m}$ and $\xmatspace{m}{d}_{<m}$, which
    simultaneously yields all $\bar{g}_i(x, a)$ modulo $f$;
  \item \label{old-algo:step6} Deduce \(g(a) = \bar{G}(x, a, a^m) \bmod f\) using
    Horner evaluation with $a^m$ modulo~$f$.
\end{enumerate}
The basis in Step~\ref{old-algo:basis} is given as a matrix
in~$\field[y]^{m\times m}$, called a \emph{relation matrix}. For any
$f$ and for a generic $a$, it has degree $d$ (see \cref{prop:genericbasisMm}),
so that its total size is~$\bigO{n^{4/3}}$ field elements.
This is also an upper bound on the size of the other intermediate
objects computed by the algorithm. In these conditions, the 
whole procedure uses $\softO{n^{(\expmm+2)/3}}$ operations in
$\field$. This can be improved slightly by taking advantage of fast
rectangular matrix products, and this forms the basis for a more complex
algorithm that handles arbitrary input by Las Vegas
randomization~\cite{NSSV24}.

\myparagraph{Modular composition with bivariate input.}
Steps~\ref{old-algo:step4.5} to \ref{old-algo:step6} above are
actually following an algorithm due to Nüsken and Ziegler to compute
\(G(x,a) \bmod f\) for a bivariate polynomial
\(G\)~\cite{NusZie04}. 
This article presents an acceleration of this bivariate
modular composition, under genericity assumptions. This is achieved
thanks to another relation matrix, this time in \(\xmatspace{\mm}{\mm}\),
for a new parameter $\mm$,
which is computed efficiently as a basis of the $\field[x]$-module
\[
  \nodule_\mm = \{p(x,y)\in\xyring_{<(\cdot,\mm)}\mid p(x,a)=0\bmod f\}.
\]
The composition is then performed as follows:
\begin{enumerate}[({3.}1),leftmargin=0.2cm,itemindent=0.7cm]
  \item\label{new-algo:baby-steps} Reduce $a^{j\mm^i}$ modulo $f$, for
    $i=0,1,2$ and $0 \le j < \mm$;
  \item\label{new-algo:step2} Compute a small-degree basis of $\nodule_\mm$, and let \(\dd\) be its degree; 
  \item \label{new-algo:Aij} Simultaneously reduce the polynomials of
  \ref{new-algo:baby-steps} by this basis to obtain
  $A_j,B_j$ in $\xyring_{<(\dd,\mm)}$, for \(0\le j<\mm\), such that
  \[
    A_{j}(x,a)=a^{j\mm}\bmod f \;\;\;\text{and}\;\;\;  B_{j}(x,a)=a^{j\mm^2}\bmod f;
  \]
\item \label{new-algo:barG} Build the polynomial $\bar{G}$ in $\field [x,y_0,y_1,y_2]_{<(\mmu,\mm,\mm,\mm)}$
  such that $G(x,y) = \bar{G}(x,y,y^\mm,y^{\mm^2})$;
\item \label{new-algo:last-step} Use $\bar{G}$ and the $A_j$'s and
  $B_j$'s of Step~\ref{new-algo:Aij} to compute
  \[
    G(x,a) = \bar{G}(x,a,a^\mm,a^{\mm^2})\bmod f
  \]
  by multiplying matrices in $\field[x,y]^{\mm\times\mm}_{<(\mmu,\mm)}$ and $\field[x,y]^{\mm\times \mm}_{<(\lceil \dd/\mm\rceil,\mm)}$.
\end{enumerate}
For any $\mm$, and for a generic $a$, we prove that the degree
$\delta$ of the relation matrix of Step~\ref{new-algo:step2} satisfies
\(\dd = \lceil n/\mm \rceil\). Choosing $\mm$ in
$\Theta(\ddelta^{1/3})$, we obtain the following result.

\begin{theorem}
  \label{theo:bivmodcomp}%
  For any $f$ in $\xring$ of degree $n>0$ and any $\mmu,\ddelta$ in $\mathbb{N}_{>0}$ such that~$md\in
  \Theta(n)$, there exists a
  non-empty Zariski-open subset $\mathcal{V}_{f,d}$ of
  $\overline{\field}{}^{n}$ such that for $(a_0,\ldots,a_{n-1})$ in $ \mathcal{V}_{f,d} \cap
  \field^n$, $a = \sum_{0 \le i < n} a_i x^i$, and any $G\in\xyring_{<(\mmu,\ddelta)}$, one can compute
  the composition $G(x,a) \bmod f$ in $\softO{\mmu\ddelta^
    {(\expmm+2)/3}}$ operations in $\field$.
\end{theorem}
In the univariate case with $(\mmu,\ddelta)= (1,n)$, this provides us
with an alternative way to recover the operation count 
$\softO{n^{(\expmm+2)/3}}$ of the algorithm in
Steps~\ref{old-algo:baby-steps}-\ref{old-algo:step6}, at least for
generic input. In both cases, the dimensions of the relation matrices
are of order $n^{1/3}$.

This type of bivariate composition
is also at the core of Nüsken and Ziegler's bivariate multipoint
evaluation \cite[Algo.~11]{NusZie04}. For this question, our
improved algorithm leads to the following.

\begin{theorem}\label{theo:multipoint}
Let $\mmu,\ddelta,n$ be in $\mathbb{N}_{>0}$ with~$md\in \Theta(n)$.
For any $\xi=(x_1,\dots,x_{n})$ in $\field^n$ with pairwise distinct~$x_i$,
there exists a
  non-empty Zariski-open subset $\mathcal{W}_{\xi,d}$ of
  $\overline{\field}{}^{n}$ such that for 
$(y_1,\dots,y_{n})$ in  $\mathcal{W}_{\xi,d} \cap
  \field^n$,  any
$G(x,y)\in\xyring_{<(\mmu,\ddelta)}$ can be evaluated at all
$(x_i,y_i)\in\field^2$, for $1\le i\le n$,
in~$\softO{\mmu\ddelta^{(\expmm+2)/3}}$ operations in~$\field$.
\end{theorem}
\cref{theo:bivmodcomp,theo:multipoint} are proved in
\S\ref{sec:bivmodcomp}. For comparison, Nüsken and Ziegler's method
has complexity $\softO{\mmu\ddelta^{\expmm_2/2}}$, with no genericity
assumption~\cite[Thm.~8]{NusZie04}. There, $\expmm_2$
is such that one can multiply an~$s\times s$ by an~$s\times s^2$
matrix in $O(s^{\omega_2})$ operations. One can take $\omega_2 =
\omega+1$ and then our algorithm has a better complexity bound.
This persists with the slightly better bounds known for~$\omega_2$ \cite{LG24,VWXXZ24,Almanetal25}.

\myparagraph{New univariate composition algorithm.}  For our 
question of computing $g(a) \bmod f$, we combine the ideas from both
preceding algorithms, using relation matrices for both modules: after
choosing our parameter $m$, we reduce $g$ by the basis of $\module_m$
to obtain a bivariate polynomial $G$ in $\xyring_{<(m,d)}$, with $d
=\lceil n/m\rceil$, then we use the algorithm of the previous
paragraph to evaluate $G(x,a) \bmod f$.

This time, the optimal choice for $m$ is $m=\lceil n^{1/4} \rceil$;
this gives $d\in\Theta(n^{3/4})$ and allows us to set the parameter
$\mm$ in the previous paragraph to $\mm=m$ (which is indeed in
$\Theta(d^{1/3})$). The resulting algorithm looks as follows:

\begin{enumerate}[({4.}1),leftmargin=0.2cm,itemindent=0.7cm]
  \item\label{newer-algo-init}  Perform Steps~\ref{new-algo:baby-steps}-\ref{new-algo:Aij}, with $\mu=m$;
  \item\label{new-algo:truncpowers} Use the polynomials $A_j$ and $B_j$ to compute the truncations
   $[x^ia^{-k-1}\rem f]_0^{m-1}$ for $0 \le i < m$ and $0 \le k < 2d$;
  \item From these, compute a basis of~$\module_m$ of degree at most \(d\);
  \item\label{new-algo:step4} Reduce~$g$ by this basis to obtain
  $G\in\xyring_{<(m,d)}$  such that $G(x,a)=g(a)$ modulo $f$;
  \item\label{new-algo:last-composition} Conclude with Steps~\ref{new-algo:barG}-\ref{new-algo:last-step}
\end{enumerate}

Step~\ref{new-algo:truncpowers} is the most technical one.  The
corresponding algorithm (in \S\ref{sec:trunc_powers}) improves on the
one used in~\cite{NSSV24} by making use of the polynomials $A_j$ and
$B_j$ introduced in the previous paragraph, and manipulates objects of
size $\bigO{mn}$ using $\softO{m^{\expmm-1}n}$ arithmetic operations.
This is the same
complexity bound as the final composition in
Step~\ref{new-algo:last-composition}, in the same way as
Steps~\ref{old-algo:step2} and~\ref{old-algo:step5} had the same
complexity bound in our previous algorithm. These apparent coincidences are
explained by complexity equivalences, discussed in
\S\ref{sec:transpositions}. Altogether, we obtain the following
result
proved in \S\ref{sec:composition}.
\begin{theorem}
  \label{thm:modular_composition}%
  For any $f$ in $\field[x]$ of degree
  $n>0$, there exists a non-empty
  Zariski-open subset $\mathcal O_{f}$ of $\overline{\field}{}^{n}$ such
  that, for $(a_0,\ldots,a_{n-1})$ in $\mathcal O_{f} \cap \field^
  {n}$, $a = \sum_{0 \le i < n} a_i x^i$,
  and any $g \in \field[x]_{<n}$, one can compute
$g(a) \bmod f$
  using $\softO{n^{(\expmm+3)/4}}\subset \bigO{n^{1.343}}$
  operations in
  $\field$.
\end{theorem}
This complexity exponent satisfies
\[5/4\le ({\expmm+3})/4\le 1.343\le1.452\le 3/2,\]
where the first approximation is obtained with the best known bound
on~$\omega$ (it improves upon the previous~$\bigO{n^{1.43}}$~\cite{NSSV24}); the second one uses
Strassen's algorithm; the last one shows that for the first time, 
$3/2$
is reached even without fast matrix multiplication. 

Along computations, the algorithm is able to detect whether
genericity assumptions
are met, and throw an error otherwise. Then one can
fall back on other algorithms. It
is natural to wonder whether randomization could help and yield a Las Vegas
algorithm, as it did in~\cite{NSSV24}; we leave this to future work.

%%%%%%%%%%%%%%%%%%%%%%%%%%%%%%%%%%%%%%%%%%%%%%%%%%%%%%%%%%%%%%%%%%%%%%%%%%%%%%%%%%%
%
%              SEC: MODULES AND DIVISION
%
%%%%%%%%%%%%%%%%%%%%%%%%%%%%%%%%%%%%%%%%%%%%%%%%%%%%%%%%%%%%%%%%%%%%%%%%%%%%%%%%%%%

\section{Modules of relations and their bases}
\label{sec:modules}

Let \(a, f \in \xring\) be polynomials with \(f\neq0\) and let
$n=\deg f$. The elements of the bivariate ideal \(\ideal = \idealGens
\subseteq \xyring\) are called \emph{relations} in what follows, and
play an important role in our algorithms.  In particular, we make use
of relations of small degrees in $x$ or $y$, by focusing on the
$\field[y]$- and $\field[x]$-modules
\begin{align}\label{eq:Mm}
  \module_\mm
  % = \ideal \cap \xyring_{<(\mm,\cdot)}
  &= \ideal_{<(\mm,\cdot)}
  = \{p \in \xyring_{<(\mm,\cdot)} \mid p(x,a) = 0 \bmod f\},\\
  \nodule_\mm\label{eq:Nm}
  % = \ideal \cap \xyring_{<(\cdot,\mm)}
  &= \ideal_{<(\cdot,\mm)}
  = \{p \in \xyring_{<(\cdot,\mm)} \mid p(x,a) = 0 \bmod f\}.
\end{align}
Both are free and of rank~$\mm$. Their bases can be represented by nonsingular
matrices in $\field[y]^{\mm\times \mm}$ and $\field[x]^{\mm\times \mm}$, respectively,
each column containing the coefficients of a relation. 

In this section, we show that for generic~\(a\), both \(\nodule_\mm\) and~\(\module_\mm\)
have bases of degree $\lceil n/\mm\rceil$
(\cref{lem:generic_bases_Nm,prop:genericbasisMm}), and we present
algorithms computing such bases efficiently.  We also show how a
reduction by these bases allows one to compute the polynomials
$A_j(x,y)$ and $B_j(x,y)$ for Step~\ref{new-algo:Aij}, as well as the
polynomial $G$ for Step~\ref{new-algo:step4} in \cref{sec:intro}
(\cref{cor:AjBj,lem:divrem_y}).

\subsection{Minimal bases}
\label{sec:modules:relbas}

When dealing with free modules over a polynomial ring in one variable,
one usually computes with bases (which we see as univariate
polynomial matrices) having specific forms, to benefit from their
smaller size and stronger properties. We first recall basic
definitions and properties (see also \cite{Kailath80,BeLaVi99}).

\begin{definition}
  Let $P\in\field[x]^{\mm\times \mm}$ be a polynomial matrix with no zero
  column and let $d_j$ be the maximal degree of the entries in its
  column \(j\), for $0 \le j < \mm$. We write \(\deg(P)\)
  for \(\max_{0 \le j < \mm} d_j\).
  The matrix \(P\) is \emph{column reduced} when $\deg(\det(P)) =
  d_0+\dots+d_{\mm-1}$. A basis of a free
  $\field[x]$-module is \emph{minimal} if it is column reduced.
\end{definition}

This terminology is explained by the fact that the column 
degrees of {column reduced} bases are lexicographically {minimal}
among all bases of the module, up to permuting the columns to make these
degrees nondecreasing.
Minimal bases are not unique and the algorithms we use produce them in 
variants of Popov forms.

\begin{definition}
  A matrix $P=(p_{ij})\in\field[x]^{\mm\times \mm}$ with no zero column is in 
  \emph{weak Popov form} if for all~$i\neq j$,
\begin{equation}\tag{$C_1$}\label{eq:cond_row}
\deg(p_{ij})\le \deg(p_{jj}),\quad\text{with strict inequality if
$i>j$.}
\end{equation}
It is in
\emph{Popov form} if moreover, for all $i$, 
\begin{equation}\tag{$C_2$}\label{eq:cond_col}
p_{ii}\text{ is monic and }\deg(p_{ii})>\deg(p_{ij}),\quad i\neq j.
\end{equation}
Finally, for a \emph{shift}
\(\bm{s} = (s_0,\ldots,s_{\mm-1}) \in \ZZ^\mm\), the matrix is in
\emph{$\bm{s}$-Popov form} if it satisfies \eqref{eq:cond_col},
while \eqref{eq:cond_row} is replaced by
\[
\deg(p_{ij})+s_i\le \deg(p_{jj})+s_j,\quad\text{with strict inequality
if
$i>j$.}
\]
For any given $\bm{s}$, a \(\xring\)-submodule of \(\xring^\mm\) of
rank \(\mm\) admits a unique basis in \(\bm{s}\)-Popov form.
\end{definition}

\subsection{Computing modulo \texorpdfstring{\(\xring\)}{K[x]}-relations}
\label{sec:modules:divrem_x}

For \(\mm\ge 2\), \(\nodule_\mm\) has a basis
\(\{f,y-a\rem f, \dots,y^{\mm-1}-a^{\mm-1}\rem f\}\),
which is not minimal. Its matrix is triangular, with
determinant $f$. Thus,
all bases of \(\nodule_\mm\) have determinant \(\lambda f\) for some \(\lambda \in \field \setminus \{0\}\).
It turns out that a minimal basis of \(\nodule_\mm\) can be computed
efficiently, along with reductions with respect to this basis.

\begin{proposition}[Minimal basis {\upshape\ulcshape\&} \textsc{reduction}]
  \label{prop:Kx_basis_and_quorem}%
  Let \(n \ge \mm > 0\), let \(f\in\xring\) have degree \(n\), and let \(a \in
  \xring_{<n}\). The minimal basis in Popov form~$R\in \xmatspace{\mm}{\mm}$ of
  \(\nodule_\mm\) from \cref{eq:Nm} can be computed
  using \(\softO{\mm^{\expmm-1} n}\) operations in~\(\field\) and
  satisfies \(\lceil n/\mm \rceil \le \deg(R) \le n\).
  Given \(\ell\) polynomials \(u_0, \ldots, u_{\ell-1}\) in \(\xring_ {<n}\),
   using \(\softO{\lceil \ell / \mm \rceil \mm^{\expmm-1} n}\)
  operations in \(\field\)
  one can compute polynomials \(U_0, \ldots, U_{\ell-1}\) in
  \(\xyring_{<(\dd,\mm)}\) such that \(U_j(x,a) = u_j \bmod f\), where
  \(\dd = \deg(R)\).
\end{proposition}

\begin{proof}
  Let \(\bar{a}_i = a^i \rem f\), \(0 \le i < \mm\). The Popov basis~$R$
  (with shift \((0,\ldots,0)\)) for the module $\nodule_\mm$ of solutions of the linear
  system
  \begin{equation}\label{eq:ap_eq_0_mod_f}
    \bar{a}_0p_0+\dots+\bar{a}_{\mm-1}p_{\mm-1}=0\bmod f
  \end{equation}
  can be computed in \(\softO{\mm^{\expmm-1} n}\) operations in \(\field\)
  \cite[Thm.\,1.4]{Neiger16}. The degree \(\dd = \deg(R)\) satisfies \(\lceil
  n/\mm \rceil \le \dd \le n\), since the sum of the column degrees of \(R\) is
  exactly \(\deg(\det(R)) = \deg(f) = n\).

  In fact, both \(R\) and \(U_0,\ldots,U_{\ell-1}\) are computed in
  \(\softO{(\mm+\ell)^{\expmm-1} n}\) operations in \(\field\), using the same
  algorithm with the equation
  \begin{equation}\label{eq:division_mod_f}
    \bar{a}_0p_0+\dots+\bar{a}_{\mm-1}p_{\mm-1} - u_0q_0-\dots-u_{\ell-1}q_{\ell-1}=0\bmod f
  \end{equation}
  and the shift \(\bm{s} = (0,\ldots,0,n,\ldots,n) \in \ZZ^{\mm+\ell}\),
  where \(0\) appears \(\mm\) times and \(n\) appears \(\ell\) times,
  the unknowns being the $p_i$'s and~$q_j$'s.
  Indeed, writing the corresponding basis as
  \[
    P =
    \begin{bmatrix}
      P_{0} & P_{1} \\
      Q_{0} & Q_{1}
    \end{bmatrix}
    \in \xmatspace{(\mm+\ell)}{(\mm+\ell)},
  \]
  where \(P_{0}\) is \(\mm \times \mm\) and \(Q_{1}\) is \(\ell \times \ell\),
  we now show that \(P_{0} = R\), \(Q_{0} = 0\), \(Q_{1} = I_\ell\)
  and, most importantly, \(P_{1} = [U_{ij}]\), with \(U_j = \sum_{0 \le i <
  \mm} U_{ij} y^i\).

  Note first that, by definition of \(\bm{s}\)-Popov forms, the principal
  submatrices \(P_{0}\) and \(Q_{1}\) are in \((0,\ldots,0)\)- and
  \((n,\ldots,n)\)-Popov form, respectively. Also, \cref{eq:cond_col}
  implies \(\deg(P) \le \deg(\det(P))\), while this is at most $\deg(f) =
  n$, by the same reasoning as for~$\nodule_\mm$ above.
  Thus, the choice of shift \(\bm{s}\) and the definition of
  \(\bm{s}\)-Popov form ensure that \(\deg(Q_{0}) + n < \deg(P_{0}) \le
  n\), hence \(Q_{0} = 0\). This implies that \(P_{0}\) is the
  \((0,\ldots,0)\)-Popov basis of the solutions of 
  \cref{eq:ap_eq_0_mod_f}, hence \(P_{0} = R\) by uniqueness.

  Next, since both \(\det(P)\) and \(\det(R)\) are equal to \(f\) up to
  multiplication by a nonzero constant, \(\det(Q_{1})\) is itself a constant,
  and thus \(Q_{1} = I_{\ell}\), since the identity matrix is the unique
  unimodular matrix in \((n,\ldots,n)\)-Popov form.

  As a result, writing \(P_{1} = [\bar{U}_{ij}]\) for the entries of
  \(P_{1}\), the fact that the \((\mm+j)\)-th column of \(P\) is a
  solution of \cref{eq:division_mod_f} implies that
  \[
    \sum_{i < \mm} \bar a_i \bar{U}_{ij} = u_j \bmod f,
    \text{ that is},
    \sum_{i < \mm} \bar{U}_{ij} a^i = u_{j} \bmod f.
  \]
  Finally, \(\deg(\bar{U}_{ij}) < \deg(R_{ii})\) follows from the fact that
  \(P\) is in \(\bm{s}\)-Popov form. By uniqueness of \(P\), there is only one
  set of polynomials \(\bar{U}_{ij}\) satisfying these degree bounds and
  \cref{eq:division_mod_f}, hence \(\bar{U}_{ij} = U_{ij}\).

  The claimed cost bound comes from applying the above method on 
  \(\lceil \ell / \mm \rceil\) lists, each containing \(\le \mm\) of the
  polynomials \(u_i\).
\end{proof}

The composition algorithm, and its components in
\cref{sec:bivmodcomp,sec:trunc_powers}, use the following consequence of
\cref{prop:Kx_basis_and_quorem}.

\begin{corollary}
  \label{cor:AjBj}%
  Let \(n \ge \mm > 0\), let \(f\in\xring\) have degree \(n\), let
  $a\in\xring_{<n}$, and let $\dd$ be the degree of a minimal basis of
  $\nodule_{\mm}$. Using $\softO{\mm^{\expmm-1}n}$ operations
  in~$\field$, one can compute polynomials $A_j,B_j$
  in~$\xyring_{<(\dd,\mm)}$ for $0\le j<\mm$, such that
  $A_{j}(x,a)=a^{j\mm}\bmod f$ and $B_{j}(x,a)=a^{j\mm^2}\bmod f$ for
  all \(j\).
\end{corollary}
\begin{proof}
  By \cref{prop:Kx_basis_and_quorem}, the sought \(\ell = 2\mm\) polynomials are
  found in $\softO{\mm^{\expmm-1}n}$ from
  $a_j=a^{j\mm}\rem f$ and $a_{\mm+j}=a^{j\mm^2}\rem f$ for $j=0,\dots,\mm-1$,
  themselves computed iteratively in $\softO{\mm n}$ ops.
\end{proof}

The above results come from deterministic algorithms which support any \(a\)
and \(f\). In modular composition, we use the following bound that
holds generically. The proof is in \cref{appendix:genericity}.
\begin{lemma}
  \label{lem:generic_bases_Nm}%
  Let \(n \ge \mm > 0\) and \(f\in\xring\) have degree \(n\). For a generic  $a \in \xring_{<n}$,
 any minimal basis of $\nodule_\mm$ has degree $\lceil n/\mm\rceil$.
\end{lemma}

The genericity is characterized by a nonzero polynomial $\witnessN_{f,\mm}\in \field[\bar a_0,\ldots \bar a_{n-1}]$ 
(the $\bar a_i$’s are new
indeterminates) whose zero set must be avoided by the coefficients of
$a$; it is given in \cref{appendix:genericity}. 
The detection of non-generic input is straightforward by inspecting the
degrees in a minimal basis of $\nodule_\mm$, whose efficient
computation is independent from genericity aspects, as seen above.

\subsection{Computing modulo \texorpdfstring{\(\yring\)}{K[y]}-relations}
\label{sec:modules:divrem_y}

The module~$\module_\mm$ from \cref{eq:Mm} is more delicate to compute
with. We rely on our previous work on this module~\cite{NSSV24}.
Minimal bases of~$\module _\mm$ have degree $\lceil n/\mm\rceil$ for
generic $a$ (see \cref{appendix:genericityMm}).
However, for their computation, we exploit a truncation
technique whose success requires additional genericity conditions.

\begin{proposition}[Minimal basis {\cite[Algo.\,5.1 and \S\,7.3.2]
{NSSV24}}]%
\label{prop:genericbasisMm}
Let \(n \ge \mm > 0\) and $f \in \xring$ have degree $n > 0$ with $f(0)\neq 0$.
 There is an algorithm
  which takes as input any $a \in \xring_{<n}$ with $\gcd(a,f)=1$ along with
  the truncated powers
  \[
    [x^{i}a^{-k-1}\rem f]_0^{\mm-1} \quad\text{for } 0 \le i < \mm \text{ and } 0 \le k < 2\dd,
  \]
  and, using $\softO{\mm^{\expmm-1}n}$ operations in \(\field\), either returns a
  minimal basis of \(\module_\mm\) of degree $\dd = \lceil n/\mm\rceil$  for generic $a$ or detects that $a$ does not
  satisfy the genericity condition.
\end{proposition}
\begin{proof}
  Algorithm~5.1 of~\cite{NSSV24} first computes the truncated powers
  of the proposition. Next, its Steps~2 and~3 compute a weak Popov
  matrix \(R \in \ymatspace{\mm}{\mm}\) of degree at most~$2\dd$ in $\softO
  {\mm^{\expmm-1}n}$ operations in $\field$~\cite[Prop.\,5.6 and its
  proof]{NSSV24}.
  For a generic
  $a$ and when $f(0)\neq 0$, this matrix~$R$ is a basis
  of~$\module_\mm$ with determinant of degree~$n$~\cite[\S\,7.3.2,
  p.\,46]{NSSV24}. It follows that \(\deg(R)=\lceil n/\mm
  \rceil\)~\cite[Prop.\,5.6]{NSSV24}.

  The genericity in $a(x)$ is precisely characterized by a polynomial
  $\witnessM_{f,\mm}$ \cite[Prop.\,7.6]{NSSV24} and by the requirement
  $\gcd(a,f)=1$, which can be characterized as the non-cancellation of the
  resultant $\operatorname{Res}(a,f)$. The detection of non-generic
  input is performed algorithmically by the flag {\sc NoCert} in
  Algo.\,5.1 of \cite{NSSV24}.
\end{proof}

For the reduction of a polynomial in \(\xyring_{<(\mm,\cdot)}\)
with respect to this basis, we
rely on a minimal nullspace basis computation~\cite{ZLS12}.

\begin{lemma}[Reduction]\label{lem:divrem_y}%
  %% Note that here, we do not require R to be a basis, can be a right-multiple of some basis
  Given a nonsingular \(R\) in \(\ymatspace{\mm}{\mm}\) of degree \(\le \dd\) and
  whose columns are relations of \(\module_\mm\), and given a polynomial
  \(g(x,y)\) in \(\xyring_{<(\mm,\cdot)}\) of \(y\)-degree in \(\bigO{n}\), one
  can compute a polynomial \(G(x,y)\) in \(\xyring_{<(\mm,\dd)}\) such that
  \(G(x,a) = g(x,a) \bmod f\),
  using \(\softO{\mm^\expmm (\dd + n/\mm)}\) operations in \(\field\).
\end{lemma}
For a detailed algorithm in the specific case of this lemma, and
a proof of the 
complexity bound, we refer to \cite[Sec.\,4.2]{NSSV24}.

%%%%%%%%%%%%%%%%%%%%%%%%%%%%%%%%%%%%%%%%%%%%%%%%%%%%%%%%%%%%%%%%%%%%%%%%%%%%%%%%%%%
%
%              BIVARIATE MODULAR COMPOSITION
%
%%%%%%%%%%%%%%%%%%%%%%%%%%%%%%%%%%%%%%%%%%%%%%%%%%%%%%%%%%%%%%%%%%%%%%%%%%%%%%%%%%%

\section{Bivariate Modular Composition}
\label{sec:bivmodcomp}

In this section, we are given $f$ in $\xring$ of degree $n$, $a$ in
$\xring_{<n}$, and $G\in\xyring_{<(\mmu,\ddelta)}$ for some positive
integers $\mmu, \ddelta$, and we prove the following proposition on the
cost of computing $G(x,a)\bmod f$. \cref{theo:bivmodcomp} will easily
follow (in that theorem, we suppose that $\mmu \ddelta$ is $\Theta(n)$,
and we assume $a$ is generic).

\begin{proposition}
  \label{prop:bivmodcomp}%
  Consider integers $n,\mmu,d$ and $\mm$, with $d^{1/3} \le \mm \le
  n$.  Given $f$ in $\xring$ of degree $n$, $a$ in $\xring_{<n}$,
  $G\in\xyring_{<(\mmu,\ddelta)}$, and the polynomials $(A_j)_{0 \le j
    < \mm}$ and $(B_j)_{0 \le j < \mm}$ from \cref{cor:AjBj}, one can
  compute the composition $G(x,a)\bmod f$ in
  $\softO{\mm^{\expmm}(\dd+\mmu \mm)}$ operations in $\field$, where
  $\dd$ is the degree of a minimal basis of $\nodule_{\mm}$.
\end{proposition}

\subsection{Algorithm and proof of Prop.~\ref{prop:bivmodcomp} \texorpdfstring{\&}{and} Thm.~\ref{theo:bivmodcomp}.}
\label{sec:bivmodcomp:algo}

The principle is to use an inverse Kronecker substitution. Since
$\ddelta \le \mm^3$, we can write any index $i < \ddelta$ as
$i=i_0+i_1\mm+i_2\mm^{2}$ with all $0\le i_j<\mm$, and then map the
monomial $y^i$ to $y_0^{i_0}y_1^{i_1}y_2^{i_2}$. By linearity this
gives a map
\begin{align*}
  \xyring_{<(\mmu,\ddelta)} & \to     \field[x,y_0,y_1,y_2]_{<(\mmu,\mm,\mm,\mm)} \\
  G(x,y)                  & \mapsto \bar{G}(x,y_0,y_1,y_2).
\end{align*}
The evaluation at \(y = a \bmod f\) amounts to computing modulo the
collection of relations \(y^j = a^j \bmod f\) for \(0 \le j <
\ddelta\), in the \(\xring\)-module
\(\xyring_{<(\cdot,\ddelta)}\). Through the above map, 
for \(0 \le j < \mm\),
these relations become
\[
  y_0^j = a^j \bmod f;\ 
  y_1^j = A_j(x,y_0) \bmod f;\ 
  y_2^j = B_j(x,y_0) \bmod f
\]
that we use to compute
$S(x,y_0) = \bar{G}(x,y_0,a^\mm,a^{\mm^2}) \bmod f$.  
\begin{description}[leftmargin=0cm,labelindent=0.15cm,font=\mdseries\itshape]
\item[Step 1: computing partial sums.] \!%
Writing the
coefficients of $\bar G$ as
\[
  \bar{G}(x,y_0,y_1,y_2)
  =
  \sum_{0 \le i_1,i_2 < \mm}s_{i_1i_2}(x,y_0)y_1^{i_1}y_2^{i_2},
\]
one first computes the sums
\[
  s_{i_2}
  =
  \sum_{0 \le i_1 < \mm} s_{i_1i_2}(x,y_0)A_{i_1}
  (x,y_0)\in\xynotring_{<(\mmu+\dd,2 \mm)}, \;\; 0\le i_2<{\mm}.
\]
The simultaneous computation of the sums~$s_{i_2}$ can be expressed as
the product of the $\mm\times \mm$ matrix with entries
$s_{i_1i_2}\in\xynotring_{<(\mmu,\mm)}$ by the vector in
$\xynotring_{<(\dd,\mm)}^{\mm}$ with entries $A_{i_1}$. Using the
$y_0{}^{\lceil \dd/\mm\rceil}$-adic expansion of these polynomials, this vector
can be expanded into an $\mm\times \mm$ matrix with entries in
$\xynotring_{<(\lceil \dd/\mm\rceil,\mm)}$. This matrix product can thus be computed at
the announced cost $\softO{\mm^{\expmm}(\dd+\mmu \mm)}$.

\item[Step 2: computing $S(x,y_0)$.] The computation of 
\[
  S(x,y_0)
  =
  \sum_{0 \le i_2 < \mm} s_{i_2}(x,y_0)B_{i_2} (x,y_0)\in\xynotring_{<(\mmu+2\dd,3 \mm)}
\]
is a straightforward loop. Each summand can be computed naively in 
$\softO{\mm(\dd+\mmu)}$ operations in $\field$; the total is a negligible
$\softO{\mm^2(\dd+\mmu)}$.

\item[Step 3: Horner evaluation.]  Finally, given $S(x,y_0)$,
  Horner evaluation at $y_0=a \bmod f$ yields \(S(x,a) \bmod f\) at a
  cost of $\softO{(n + \mmu + \dd)\mm}$ operations. Since \(\dd \ge
  n/\mm\), this is $\softO{ \mm^2\dd + \mmu\mm}$, which is negligible
  compared to the cost of the other operations.
\end{description}

Altogether, the cost is $\softO{\mm^{\expmm}(\dd+\mmu
\mm)}$, which establishes \cref{prop:bivmodcomp}.

In the statement of \cref{theo:bivmodcomp}, our assumption is that
$\mmu \ddelta \in \Theta(n)$, which implies $\ddelta\le cn$ for
some constant $c$, hence $\ddelta \le n^3$ for $n$
large enough, so the previous proposition applies. In this context, we
choose $\mm=\lceil \ddelta^{1/3}\rceil \in \Theta(\ddelta^{1/3})$.  If
furthermore the coefficients of $a$ belong to the Zariski open set $
\mathcal{V}_{f,d}$ defined with the polynomial~$\Phi_{f,\lceil\ddelta^
{1/3}\rceil}$ 
from \cref{lem:generic_bases_Nm}, the degree $\dd$ is equal to $\lceil
n/\mm\rceil \in \Theta(n/d^{1/3})$.
In addition to the cost in \cref{prop:bivmodcomp}, we first need to
compute all polynomials $A_j$ and $B_j$. By \cref{cor:AjBj}, this is
$\softO{\mm^{\expmm-1}n}$ operations in $\field$, that is,
$\softO{\ddelta^{(\expmm-1)/3}n}$. This is
$\softO{\mmu\ddelta^{(\expmm+2)/3}}$, since $\mmu \ddelta \in \Theta(n)
$. The runtime for bivariate composition itself is
$\softO{\ddelta^{\expmm/3}(n/d^{1/3}+\mmu \ddelta^{1/3})}$ operations
in $\field$, and using again $\mmu \ddelta \in \Theta(n)$, this is
also $\softO{\ddelta^{\expmm/3}(\mmu\ddelta^{2/3}+\mmu
  \ddelta^{1/3})}=\softO{\mmu\ddelta^{(\expmm+2)/3}}$. This proves
\cref{theo:bivmodcomp}.

\subsection{Multipoint evaluation: proof of Thm.~\ref{theo:multipoint}}

\subsubsection*{Proof.}
It follows the steps of Nüsken and Ziegler's algorithm.  From 
 $(x_1,y_1),\dots,(x_n,y_n)$, we first compute $f=\prod(x-x_i)$ and the interpolating
polynomial~$a\in\xring_{<n}$ such that~$a(x_i)=y_i$ for all~$i$, both
in~$\softO{n}$ operations in~$\field$~\cite{GaGe99}.
Given $G$ in $\field[x,y]_{<(\mmu,\ddelta)}$ and from \cref{theo:bivmodcomp}, if the coefficients
of~$a$ 
are in $\mathcal{V}_{f,d} \cap   \field^n$, $h = G(x,a)\bmod f$ is obtained
in~$\softO{\mmu\ddelta^{(\expmm+2)/3}}$ operations in \(\field\).
The desired values are then deduced by
univariate multipoint evaluation of \(h\) at~$x_1,\dots,x_ {n}$, again
in $\softO{n}$ operations in~$\field$.
The points~$(y_1,\dots,y_{n}) \in \field^n$ 
for which the coefficients 
of~$a$ 
are in $\mathcal{V}_{f,d} \cap   \field^n$ are in the image~$\mathcal{W}_{\xi,d}$ of $\mathcal{V}_{f,d}$ by the invertible Vandermonde matrix constructed from~$
\xi=(x_1,\dots,x_{n})$, which is
thus also a non-empty Zariski open set.

\subsubsection*{Arbitrary sets of points}
\Cref{theo:multipoint} and Nüsken and Ziegler's results hold for
points with pairwise distinct $x$-coordinate.
For arbitrary sets of points, a natural approach
(going back to~\cite{NusZie04}) is to apply a random change of
coordinates, since with high probability, this guarantees that all
$x$-coordinates are distinct. 
However, we do not expect this to be
sufficient to have small-degree relation matrices needed for the target complexity bound, 
and further work will be needed here as well to obtain a
Las Vegas algorithm.

When~$m=d=\sqrt n$, Nüsken and Ziegler's algorithm has
complexity $\softO{n^{(\expmm_2+2)/4}}$, which is in $\bigO{n^{1.313}}$, while 
\cref{theo:multipoint} gives $\softO {n^{(\expmm+5)/6}}$, which is in $\bigO{n^
{1.229}}$. 
Without changing variables, the best result we are aware of for
arbitrary inputs is in~\cite{HoeLec25}, with a runtime of
$\softO{n^{2-2/(\expmm+1)}}\subset \bigO{n^{1.41}}$ operations in
$\field$.
Note that several other
results are known in an \emph{amortized}
model~\cite{NRS20,HoeLec21b,HoeLec21c,HoeLec23,HoeLec25}, some of them
also requiring generic inputs.

%%%%%%%%%%%%%%%%%%%%%%%%%%%%%%%%%%%%%%%%%%%%%%%%%%%%%%%%%%%%%%%%%%%%%%%%%%%%%%%%%%%
%
%              TRUNCATED POWERS
%
%%%%%%%%%%%%%%%%%%%%%%%%%%%%%%%%%%%%%%%%%%%%%%%%%%%%%%%%%%%%%%%%%%%%%%%%%%%%%%%%%%%

\section{Truncated powers}
\label{sec:trunc_powers}

\Cref{new-algo:truncpowers} of the new composition algorithm computes the
truncated powers $[x^ia^{-k}\rem f]_0^{m-1}$ for $0\le i<m,0\le k\le
2d$. Using the
small-degree polynomials \(A_j\) and \(B_j\) in \cref{new-algo:Aij}
lets us improve the previous result for this
computation~\cite [Prop.\,3.6] {NSSV24}.

\begin{proposition}
  \label{prop:truncatedpowers}%
  Consider integers $n,\mmu,d$ and $\mm$ with $d^{1/3} \le \mm \le
  n$.  Given $f$ in $\xring$ of degree $n$, $a$ in $\xring_{<n}$,
  and the polynomials $(A_j)_{0 \le j
    < \mm}$ and $(B_j)_{0 \le j < \mm}$ from \cref{cor:AjBj}, one
  can compute the truncations \([b a^k\rem f]_0^{\mmu-1}\) for \(0\le
  k<\ddelta\) using $\softO{\mm^\expmm (\dd + \mmu \mm)}$ operations
  in~$\field$, where \(\dd\) is the degree of a minimal basis of
  \(\nodule_\mm\).
\end{proposition}

In this section, we use the bracket notation $[\cdot]_0^{\mmu-1}$ with
bivariate polynomials: it stands for the first coefficients $\mmu$
with respect to $x$. Also, while the proposition
computes truncations of terms $b a^k\rem f$, we will apply this with
$a$ replaced by its inverse modulo $f$.

\subsection{Simultaneous truncated products}

The first ingredient is the simultaneous computation of several truncated
products with a specific structure.

\begin{lemma}\label{lem:simultruncprod}%
  Let \(n, \mm, \hat{\mm}, \dd, \mmu \in \NN_{>0}\) with \(\hat{\mm} \in \bigO{\mm^2}\)
  and \(\dd \le n\). Given $f$ in $\xring$ of degree $n$, \(h\) in
  \(\xyring_{<(\dd,\hat{\mm})}\), and $\eta_1,\ldots,\eta_\mm$ in $\xyring_{<(\dd,\mm)}$,
  one can compute all truncated products
  \[
    R_j=[x^{n-\dd} h \eta_j \rem f]_0^{\mmu-1} \in \xyring_{<(\mmu,\mm+\hat{\mm}-1)},
    \quad 1 \le j \le \mm
  \]
  using $\softO{\mm^\expmm (\dd + \mmu \mm)}$ operations in $\field$.
\end{lemma}

\begin{proof}
  Use segmentation with respect to $y$ to write
  \[
    h(x,y)=\sum_{0 \le i < \bar{\mm}} h_i(x,y) y^{\mm i},
  \]
  with $h_0,\ldots,h_{\bar{\mm}-1} \in \xyring_{<(\dd,\mm)}$ and $\bar{\mm} \in
  \bigO{\mm}$.
  Set $H_i=x^{n-\dd} h_i$. It is enough to compute, for all \(0 \le i < \bar{\mm}\)
  and \(1 \le j \le \mm\),
  \[
    R_{ij} = [H_i \eta_j \rem f]_0^{\mmu-1} \in \xyring_{<(\mmu,2\mm-1)}.
  \]
  Indeed, all $R_j = \sum_i R_{ij} y^{\mm i}$ can then be deduced in linear time
  with respect to the total size of the $R_{ij}$'s, which is in \(\bigO{\mmu
  \mm^3}\). Since all products $R_{ij}$ have $y$-degree less than $2\mm-1$, we
   compute them as polynomials in $\otherring[x]$, with $\otherring =
  \field[y]/\genBy{y^{2\mm-1}}$.

  For the computation, we use Algo.~3.6 from~\cite{NSSV24}. While the
  original presentation assumes that we work over a field, rather than
  a ring such as \(\otherring\), it is readily checked that all
  operations are performed in the ring.  A more serious issue is that
  a direct application of Lemma~3.5 in \cite{NSSV24} does not give the
  runtime we expect, as it does not take into account that $H_i \rem
  x^{n-\dd}=0$ and $\deg_x(\eta_j)<\dd$.  Thus, we revisit the steps of
  that algorithm, and analyze how these two properties improve its
  runtime.

  \begin{description}[leftmargin=0cm,labelindent=0.2cm,font=\mdseries\itshape]
    \item[Steps 2 and 3] \!%
      compute $\bar{H}_i=x^{n-1}H_i(1/x,y)$, which
      has $x$-degree less than $\dd$, and
      $\bar{\eta}_j=x^{n-1}\eta_j(1/x,y)$, which vanishes modulo
      $x^{n-\dd}$. This does not cost any arithmetic operation.

    \item[Step 3] \!%
      also computes the power series expansions
      $\gamma_j=\bar{\eta}_j/\bar{f}$ modulo $x^{n-1}$, where $\bar{f} = x^n
      f(1/x)$. Since the \(x\)-valuation of $\bar{\eta}_j$ is \(\ge n-\dd\), each
      such expansion takes $\softO \dd$ operations in $\otherring$, for a total
      of $\softO{\mm^2 \dd}$ operations in $\field$.

    \item[Step 4] \!%
      defines matrices $P_1$ and $P_2$: the $i$th row of
      \(P_1\) is obtained by segmenting the coefficients of
      $\bar{H}_i$ with respect to $x$, into slices of $x$-degree less
      than $\mmu$; $P_2$ is obtained from $P_1$ by discarding its last
      column. In the original presentation, these matrices have up to
      $\bigO{\lceil n/\mmu \rceil}$ nonzero columns, but here the
      rightmost ones are zero because $\bar{H}_i$ has $x$-degree $<
      \dd$. Discarding these columns in $P_1$ and $P_2$ leaves us with
      matrices $\bar P_1,\bar P_2$ for both, with $\bar{\mm}$ rows and
      $\bigO{\lceil \dd/\mmu \rceil}$ columns, and entries in
      $\otherring[x]_{<\mmu}$.

    \item[Step 4] \!%
      also defines matrices $Q_1$ and $Q_2$, obtained by
      segmenting the polynomials $\gamma_j$, as we did for
      $\bar{H}_i$. For valuation reasons, many bottom rows of these
      matrices are zero and can be discarded, leaving us with matrices
      $\bar Q_1, \bar Q_2$ with $\bigO{\lceil \dd/\mmu \rceil}$ rows
      and $\mm$ columns, and the same type of entries as those of
      $\bar P_1,\bar P_2$.

    \item[Step 5] \!%
      performs the matrix products $\bar P_1 \bar Q_1$ and
      $\bar P_2 \bar Q_2$, as well as $\mm \bar{\mm}$ products of
      polynomials in $\otherring[x]_{<\mmu}$; the latter take a total
      of $\softO{\mmu \mm^{3}}$ operations in $\field$. For computing
      a product \(\bar P_i \bar Q_i\), considering the two cases \(\mmu
      \mm \in \bigO{\dd}\) and \(\dd \in \bigO{\mmu \mm}\) leads to the
      complexity bounds \(\softO{\mm^\expmm \dd}\) and \(\softO{\mmu
        \mm^{\expmm+1}}\), respectively.

    \item[Step 6] \!%
      is unchanged, and amounts to $\bigO{\mm^2}$ multiplications in
      $\otherring[x]/\genBy{x^\mmu}$, that is, another $\softO{\mmu
      \mm^{3}}$ operations in $\field$.
  \end{description}

  The four terms arising in the cost are thus \(\mmu \mm^3\), \(\mm^2 \dd\), and
  either \(\mm^\expmm \dd\) or \(\mmu \mm^{\expmm+1}\). They are all in
  \(\bigO{\mm^\expmm (\dd + \mmu \mm)}\).
\end{proof}

\subsection{High part of a Euclidean remainder}

The following lemma is extracted from the proof of \cite[Lem.\,3.5]{NSSV24}.
$\otherring$ is a commutative ring.

\begin{lemma}
  \label{lem:divrem_lohi}%
  Let $f$ be monic of degree $n$ in $\otherring[x]$, let $P$ be in
  $\otherring[x]_{<n}$, and let $Q$ in $\otherring[x]_{<d}$ for some $d \in
  \{1,\ldots,n\}$. Given $Q$ and the high part $[P]_{n-t-d+1}^{t+d-2}$ for some
  $t\in \{1,\ldots,n-d+1\}$, one can compute the high part $[PQ \rem
  f]_{n-t}^{t-1}$ in $\softO{d+t}$ operations in \(\otherring\).
\end{lemma}
\begin{proof}
  Write the Euclidean division $PQ = h f +r$, thus with $r = PQ \rem f$, and
  consider the reversed polynomials given by
  \begin{align*}
    & \bar{P} = x^{n-1} P(1/x), \quad
    \bar{Q} = x^{d-1} Q(1/x), \\
    & \bar{f} =x^n f(1/x), \quad
    \bar{r} = x^{n-1} r(1/x), \quad
    \bar{h} = x^{d-2} h(1/x),
  \end{align*}
  all being in \(\otherring[x]\). By  construction, we have $\bar{P} \bar{Q} =
  \bar{h} \bar{f} + x^{d-1} \bar{r}$. Our input gives us access to $\bar{P}
  \bmod x^{d-1}$, so that \(\bar{h}\) can be obtained by power series division
  $(\bar{P}\bar{Q}/\bar{f}) \rem x^{d-1}$ in $\softO{d}$ operations. Once $h$
  is known, the high part of $r$ is obtained in $\softO{t+d}$ operations by
  \[
    [r]_{n-t}^{t-1}=\left[[P]_{n-t-d+1}^{t+d-2}Q-h[f]_{n-t-d+1}^{t+d-2}\right]_{d-1}^{t-1}.
    \qedhere
  \]
\end{proof}

\subsection{New TruncatedPowers algorithm}

We turn to the proof of \cref{prop:truncatedpowers}. For $k \ge 0$, we write
$c_k = b a^k \bmod f \in \ring$, where $\ring=\xring/\genBy{f}$. These are the
coefficients of the generating function
\[
  D(y) = \frac{b}{1-ay} \in \ring[[y]].
\]
Our goal can then be restated as computing $[c_k]_0^{\mmu-1}$ for
$k<\ddelta$, that is, $[D \bmod y^\ddelta]_0^{\mmu-1}$; this has size
$\Theta(\mmu \ddelta)$. Computing the first \(\ddelta\) coefficients
of $D$ before truncating in $x$ would yield an intermediate object of
size $\Theta(n\ddelta)$, which is too large for our target complexity.

We fix \(\mm = \lceil \ddelta^{1/3} \rceil\). In particular \(\ddelta
\le \mm^3\), and we focus on computing $[c_k]_0^{\mmu-1}$ for all \(0
\le k < \mm^3\). The key ingredient lies in the following lemma, which
shows how to obtain new coefficients $c_k$, of high indices, from
previously known ones. It makes use of the polynomials $(A_j)_{1\le j
  < \mm}$ and $(B_j)_{1\le j < \mm}$ from \cref{cor:AjBj}, and of the
\(y\)-reversed counterparts of these polynomials, also in
\(\xyring_{<(\dd,\mm)}\):
\[
  \alpha_j = y^{\mm-1}A_{j}(x,1/y)
  \quad \text{and} \quad
  \beta_j = y^{\mm-1}B_{j}(x,1/y)
\]
for \(1 \le j < \mm\). We often see \(\alpha_j\) and \(\beta_j\) as being in
\(\ring[y]_{<\mm}\), as for instance in the following lemma.

\begin{lemma}
  \label{lem:middlepartbis}%
  For $1 \le j < \mm$, and for any \(k \ge \mm-1\), the coefficient of \(y^k\) in
  \(D \alpha_j\) is  \(c_{j\mm + k - (\mm-1)}\), while that of \(D \beta_j\) is  \(c_{j\mm^2
  + k - (\mm-1)}\).
\end{lemma}
\begin{proof}
  We prove the claim for $\alpha_j$; the proof is similar for $\beta_j$. Let
  $a_{j,k}$ be the coefficient of \(y^k\) in $A_j$, so that
  $\alpha_j = a_{j,\mm-1} + \cdots + a_{j,0} y^{\mm-1}$. Then, a direct expansion
  shows that for any $k \ge \mm-1$, the coefficient of $y^k$ in the product $D
  \alpha_j$, working modulo \(f\), is
  \begin{align*}
    & b a^{k-(\mm-1)} a_{j,0} + \cdots + b a^{k-1} a_{j,\mm-2} + b a^{k} a_{j,\mm-1}  \\
    & = b a^{k-(\mm-1)} A_j(x,a)
    = b a^{j\mm+k-(\mm-1)}
    = c_{j\mm+k-(\mm-1)}.
    \qedhere
  \end{align*}
\end{proof}

Thus, if we know $D \bmod y^N$, multiplying it by $\alpha_j$ gives us all
coefficients $c_{j\mm},c_{j\mm+1},\ldots,c_{j\mm+N-1-(\mm-1)}$. The same holds for $\beta_j$,
but with a shift of $j\mm^2$ instead of \(j\mm\). Consider the truncations
\begin{align*}
  D_0 = D \bmod y^{\mm + 2(\mm-1)} = D \bmod y^{3\mm-2}, \\
  D_1 = D \bmod y^{\mm^2 + \mm-1},\quad D_2 = D \bmod y^{\mm^3}.
\end{align*}
We see \(D_0\) and \(D_1\) as being in \(\otherring_0[x]\) and
\(\otherring_1[x]\) respectively, where
\[
  \otherring_0 = \field[y]/\genBy{y^{3\mm-2}}
  \quad\text{and}\quad
  \otherring_1 = \field[y]/\genBy{y^{\mm^2+\mm-1}}.
\]
Recall that our goal is to compute \([D_2]_0^{\mmu-1}\).

According to the above remarks, multiplying $D_0$ by
$\alpha_1,\ldots,\alpha_{\mm-1}$ in $\otherring_0[x]$ gives us all
coefficients of $D_1$ (some of them are computed twice), and
multiplying $D_1$ by $\beta_1,\ldots,\beta_{\mm-1}$ in
$\otherring_1[x]$ gives us all coefficients of $D_2$ (with no repeated
calculations this time).

Still, computing \(D_1\) and \(D_2\) in this manner is too costly for
our target. Instead, our algorithm below exploits the fact
that we only seek low degree coefficients of \(D_2\).
It starts from
$D_0$ and computes the high part $[D_1]_{n-\dd}^{\dd-1}$, then deduces the
low part $[D_1]_{0}^{\mmu-1}$, and eventually $[D_2]_0^{\mmu-1}$.

\noindentparagraph{Step 1: all of $D_0$.} Each product $c_k = ba^k\bmod
f$ for \(0 \le k < 3\mm-2\) is obtained in $\softO{n}$ operations in
$\field$,
hence a total
in $\softO{\mm n}$ operations in \(\field\) for \(D_0\).

\noindentparagraph{Step 2: high part $[D_1]_{n-\dd}^{\dd-1}$.} For
this computation, it is enough to compute
all high parts $[D_0 \alpha_j]_{n-\dd}^{\dd-1}$ in $\otherring_0[x]$.  For
each of these products, we use \cref{lem:divrem_lohi} with $t=\dd$ and
input $[D_0]_{n-2\dd+1}^{2\dd-2}$. Each product uses $\softO{\dd}$
operations in $\otherring_0$, that is, in $y$-degree $\bigO{\mm}$, for a
total of $\softO{\mm^2\dd}$ operations in~$\field$.

\noindentparagraph{Step 3: low parts of $D_1$ and $D_2$.} This step computes
$[D_1]_0^{\mmu-1}$ and $[D_2]_0^{\mmu-1}$, which boils down to computing the low
parts of all products $D_0 \alpha_j$ for the former, and $D_1 \beta_j$ for the
latter. Splitting $D_0$ gives
\begin{align}
  [D_0 \alpha_j]_0^{\mmu-1} = & ~[D_0]_0^{\mmu-1}[ \alpha_j]_0^{\mmu-1}\bmod x^\mmu  \nonumber\\
                             & ~+ \left[x^{n-\dd}[D_0]_{n-\dd}^{\dd-1}\alpha_j\rem f\right]_0^{\mmu-1}\label{eq:loD1}
\end{align}
in \(\otherring_0[x]\), and similarly
\begin{align}
  [D_1 \beta_j]_0^{\mmu-1} = & ~[D_1]_0^{\mmu-1}[ \beta_j]_0^{\mmu-1}\bmod x^\mmu  \nonumber \\
                            & ~+ \left[x^{n-\dd}[D_1]_{n-\dd}^{\dd-1}\beta_j\rem f\right]_0^{\mmu-1},\label{eq:loD2}
\end{align}
in $\otherring_1[x]$. For the products in \cref{eq:loD1}, we need the low and
high parts of $D_0$, which are known. For the ones in \cref{eq:loD2}, we need
the high part of $D_1$, which we obtained in Step 2, and its low part, which we
derive from using \cref{eq:loD1} with $j = 1, \ldots, \mm-1$.

The products in \cref{eq:loD1} have lower cost than those in
\cref{eq:loD2} since the degree of \(D_0\) is smaller than that of
\(D_1\). Thus we focus on \cref{eq:loD2}, for $1 \le j < \mm$. The
first term can be computed directly using $\mm$ bivariate
multiplications modulo $(x^\mmu,y^{\mm^2+\mm-1})$, for a total of
$\softO{\mmu \mm^3}$ operations. For the second term, we use
\cref{lem:simultruncprod} with $\hat \mm=\mm^2+\mm-1$, and with a
runtime of $\softO{\mm^{\expmm} (\dd + \mmu \mm)}$.

Altogether, the cost is \(\softO{\mm n + \mm^2 \dd + \mmu \mm^3 +
\mm^\expmm (\dd + \mmu \mm)}\). Since \(\dd \ge n/\mm\) (see \cref{prop:Kx_basis_and_quorem}), this is bounded by
\(\softO{\mm^\expmm (\dd + \mmu \mm)}\).

%%%%%%%%%%%%%%%%%%%%%%%%%%%%%%%%%%%%%%%%%%%%%%%%%%%%%%%%%%%%%%%%%%%%%%%%%%%%%%%%%%%
%
%              PROOF OF MAIN THEOREMS
%
%%%%%%%%%%%%%%%%%%%%%%%%%%%%%%%%%%%%%%%%%%%%%%%%%%%%%%%%%%%%%%%%%%%%%%%%%%%%%%%%%%%

\section{Composition algorithm. Proof of 
Theorem \ref{thm:modular_composition}}
\label{sec:composition}

To prove our main result, we start with $f$ of degree $n$ in
$\field[x]$, and we first establish the claim under the assumption
$f(0)\neq0$. Our first condition on $a$ is that it satisfies
$\gcd(a,f)=1$ (this is a Zariski-generic property in the coefficients of~$a$).

We choose $m=\lceil n^{1/4}\rceil$ and $d=\lceil n/m\rceil$,
which implies both inequalities
$d^{1/3} \le m \le n$. We suppose that $a$ satisfies the genericity
conditions of \cref{lem:generic_bases_Nm} and \cref{prop:genericbasisMm} 
 for $\mm=m$, so in particular the minimal bases of
$\nodule_m$ and $\module_m$ have degree $d$.
Then, consider the following detailed presentation of the algorithm
sketched in \cref{sec:intro}:
\begin{enumerate}[leftmargin=0.2cm,itemindent=0.7cm]
\item Compute~$\tilde{a} = a^{-1}\bmod f$, which exists since~$\gcd
  (a,f)=1$;
\item\label{new-algo-comp:basis} Compute a minimal basis of $\nodule_m$, with degree $d$;
\item\label{new-algo-comp:baby-steps} Reduce $a^{jm^i}$ modulo $f$, for
  $i=0,1,2$ and $0 \le j < m$;
\item\label{new-algo-comp:reduce} Simultaneously reduce the polynomials of
  (\ref{new-algo-comp:baby-steps}) by the basis of (\ref{new-algo-comp:basis}) to obtain
  $A_j,B_j$ in
  $\xyring_{<(d,m)}$, for \(0\le j<m\), such that
  \[
  A_{j}(x,a)=a^{jm}\bmod f \;\;\;\text{and}\;\;\;  B_{j}(x,a)=a^{jm^2}\bmod f;
  \]
  \item\label{new-algo-comp:get-trunc} Use the polynomials $A_j$ and $B_j$ to compute the truncations
    $[x^i a^{-k-1}\rem f]_0^{m-1}$ for $0 \le i < m$ and $0 \le k <
    2d$;
  \item\label{new-algo-comp:getG} From these, compute a basis
    of~$\module_m$ of degree at most \(d\);
  \item\label{new-algo-comp:reduce2} Reduce~$g$ by this basis to obtain $G\in\xyring_{<(m,d)}$ such
    that $G(x,a)=g(a)\bmod f$;
\item\label{new-algo-comp:biv} Use $G$ and the $A_j$'s and $B_j$'s to compute $ g(a) \bmod f$.
\end{enumerate}

The first step takes~$\softO{n}$ operations in~$\field$. Computing the
basis of~$\nodule_m$ at \cref{new-algo-comp:basis} and deducing the
polynomials $A_j$ and $B_j$ at
\cref{new-algo-comp:baby-steps,new-algo-comp:reduce} altogether take
$\softO{m^{\expmm-1}n}$ operations, by
\cref{prop:Kx_basis_and_quorem,cor:AjBj}. With our choice of $m$, this is 
$\softO{n^{(\expmm+3)/4}}$.

To perform \cref{new-algo-comp:get-trunc}, we apply the algorithm of
\cref{prop:truncatedpowers} twice, with parameters $n$, $m'=2m-2$
(instead of $m$), $d$ and $\mu=m$ (so that $\delta=d$), and input
polynomials $f$, $\tilde a$, and $b=x^{m-1}\bmod f$,
resp. $b=x^{m-1}\tilde a^d \bmod f$. As in the proof of
\cref{theo:bivmodcomp}, the runtime is $\softO{m d^{(\expmm+2)/3}}$,
which is also $\softO{n^{(\expmm+3)/4}}$ for our choice of $m$.

Concatenating the results gives $[x^{m-1} a^{-k}\rem f]_0^{2m-2}$ for
$0\le k< 2d-1$. From this, Algo.\,3.8 and Prop.3.7 in \cite{NSSV24}
show
that we can deduce the truncations $[x^i a^{-k-1}\rem f]_0^{m-1}$,
for $0 \le i < m$ and $0 \le k < 2d$, in time $\softO{m^2 d}=\softO{n^{5/4}}$.

For \cref{new-algo-comp:getG}, thanks to $f(0)\neq0$,
\cref{prop:genericbasisMm} allows us to compute a minimal basis of
\(\module_m\) in time $\softO{m^{\expmm-1} n}=
\softO{n^{(\expmm+3)/4}}$. \cref{lem:divrem_y} then allows us to 
perform \cref{new-algo-comp:reduce2} in the same asymptotic runtime.
Finally, \cref{new-algo-comp:biv} is handled by
\cref{prop:bivmodcomp}, using $\softO{m d^{(\expmm+2)/3}}$ operations,
which is again in $\softO{n^{(\expmm+3)/4}}$.

The Zariski-open set $\mathcal{O}_{f}$ is defined from the
polynomial~$\Phi_{f,\lceil n^ {1/4}\rceil}$ from
\cref{lem:generic_bases_Nm}, the polynomial $\witnessM_{f,\lceil n^
  {1/4}\rceil}$ from \cite[Prop.\,7.6]{NSSV24} that we mentioned in
the proof of \cref{prop:genericbasisMm}, and by the requirement
$\gcd(a,f)=1$. This establishes the theorem for $f$ such that $f(0)
\ne 0$.

Now, let $f$ be arbitrary of degree $n$, and write it as $f = x^\alpha
f^*$, with $f^*$ of degree $n^*=n-\alpha$ and $f^*(0) \ne 0$; let also
$\bar a_0,\dots,\bar a_{n-1}$ be indeterminates (that represent the
coefficients of $a$).

Let $\Delta$ be a polynomial in $\field[\bar b_0,\dots,\bar
  b_{n^*-1}]$ that defines the complement of $\mathcal{O}_{f^*}$, and
set $\Delta^*(\bar a_0,\dots,\bar a_{n-1})=\Delta(r_0,\dots,r_
{n^*-1})$, where
$r_0,\dots,r_{n^*-1}$ are the coefficients of the remainder of $\bar
a_0 + \cdots + \bar a_{n-1}x^{n-1}$ by $f^*$. In view of the equality
$\Delta^*(\bar a_0,\dots,\bar a_{n^*-1},0,\dots,0)=\Delta(\bar
a_0,\dots,\bar a_{n^*-1})$, we deduce that $\Delta^*$ is
nonzero.

Suppose then that the coefficients of $a$ do not cancel $\Delta^*$. To
compute $h=g(a) \bmod f$, it is enough to compute $\hat a = a \bmod
x^\alpha$ and $a^* = a\bmod f^*$, then $\hat h = g(\hat a) \bmod
x^\alpha$ and $h^* = g(a^*) \bmod f^*$, and finally recover $h$ by
Chinese Remaindering. Computing $\hat a$, $a^*$ and Chinese
Remaindering, take quasi-linear time in $n$, so it remains to discuss
the cost of computing $\hat h$ and $h^*$. In both cases, we use
segmentation on $g$.
\begin{description}[leftmargin=0cm,labelindent=0.2cm,font=\mdseries\itshape]
\item[Computation of $\hat h$.] Write $g = \sum_{i \le \lceil
n/\alpha\rceil} g_i (y)
  y^{i \alpha}$, with all $g_i$ of degree less than $\alpha$.  Using
  the Kinoshita-Li algorithm~\cite{KinoshitaLi2024}, we can compute
  all $g_i(\hat a) \bmod x^\alpha$ in $\softO{\alpha
    (n/\alpha)}=\softO n$ operations. Then, after the additional
  computation of $\hat a^\alpha \bmod x^\alpha$, we can use Horner's
  scheme to recover $\hat h=g(\hat a) \bmod x^\alpha$ in softly linear
  time in $n$.
\item[Computation of $h^*$.] The same procedure applies. Since
the coefficients of
  $a$ do not cancel $\Delta^*$, the coefficients of $a^*$ do not
  cancel $\Delta$, so a single composition by $a^*$ modulo $f^*$ takes
  $\softO {{n^*}^{(\expmm+3)/4}}$ operations. The rest of the analysis
  is similar, and shows that we can obtain $h^*$ in time $\softO
  {n^{(\expmm+3)/4}}$.
\end{description}

%%%%%%%%%%%%%%%%%%%%%%%%%%%%%%%%%%%%%%%%%%%%%%%%%%%%%%%%%%%%%%%%%%%%%%%%%%%%%%%%%%%
%
%              TRANSPOSITIONS OF PROBLEMS
%
%%%%%%%%%%%%%%%%%%%%%%%%%%%%%%%%%%%%%%%%%%%%%%%%%%%%%%%%%%%%%%%%%%%%%%%%%%%%%%%%%%%

\section{Complexity equivalences}
\label{sec:transpositions}

The problem of modular composition is closely related to those of
power projections and characteristic polynomials in quotient
algebras. See, in particular, \cite{Shoup94,Shoup99} and
\cite[Sec.\,6]{Kal00-2}.  This results in equivalences between some of
the subproblems addressed here that could be used to develop
alternative algorithms.
First, \cref{prop:truncatedpowers,prop:bivmodcomp} show that two
subproblems have the same complexity bound.
In fact, there is a general relation between them:
\cref{sec:truncated_vs_bivariate} shows that, in terms of straight-line programs,
the problems of truncated powers and bivariate composition are essentially equivalent, when~$f(0)\neq 0$.
This is shown on the right side of \cref{fig:subproblems}.
Next, \cref{sec:charpoly_vs_composition}
explains that, in generic cases,  bivariate composition
could be replaced by computations of specific characteristic polynomials obtained from minimal bases.
This provides another alternative path for composition, shown on the left side of \cref{fig:subproblems}.

\begin{figure}[b]
\centering
\begin{tikzpicture}[
    box/.style={
        draw,
        rounded corners=4pt,
        minimum width=1.5cm,
        minimum height=0.7cm,
        align=center
    },
    altpath0/.style={->, shorten >=1.5pt, shorten <=1.5pt, dashed},
    altpath/.style={<->, shorten >=1.5pt, shorten <=1.5pt, dashed}
]

\node[box] (A) at (0,2)   {\footnotesize Basis of $\nodule_m$};
\node[box] (B) at (0,0.9)   {{\footnotesize Truncated powers}};
\node[box] (C) at (0,-0.2)     {\footnotesize Basis of $\module_m$};
\node[box] (D) at (0,-1.3)  {\footnotesize Bivariate composition};
\node[box] (E) at (0,-2.4)  {\footnotesize Composition};

\node[box] (F) at (2.4,-0.15)  {{\footnotesize Transposition}};
\node[box] (G) at (-2.5,-0.6) {\footnotesize Characteristic \\[-0.1cm] \footnotesize polynomial};
\node[box] (H) at (-2.5,-1.8)   {\footnotesize Inverse \\[-0.1cm] \footnotesize composition};

\draw[->, line width=0.5pt, shorten >=1.5pt, shorten <=1.5pt] (A) -- (B);
\draw[->, line width=0.5pt, shorten >=1.5pt, shorten <=1.5pt] (B) -- (C);
\draw[->, line width=0.5pt, shorten >=1.5pt, shorten <=1.5pt] (C) -- (D);
\draw[->, line width=0.5pt, shorten >=1.5pt, shorten <=1.5pt] (D) -- (E);

\draw[altpath] (B.east) .. controls +(1,0) and +(0,0.5) .. (F.north);
\draw[altpath] (F.south) .. controls +(0,-0.5) and +(1,0) .. (D.east);

\draw[altpath0] (C.west) .. controls +(-1,0) and +(0,0.5) .. (G.north);
\draw[altpath0] (G.south) -- (H.north);
\draw[altpath0] (H.south) .. controls +(0,-0.5) and +(-1,0) .. (E.west);

\end{tikzpicture}

\caption{\textmd{Subproblems involved in the new modular composition algorithm. The dotted lines indicate possible alternative paths.}
\Description{Subproblems involved in the new modular composition algorithm. The dotted lines indicate possible alternative paths.}
\label{fig:subproblems}}
\end{figure}
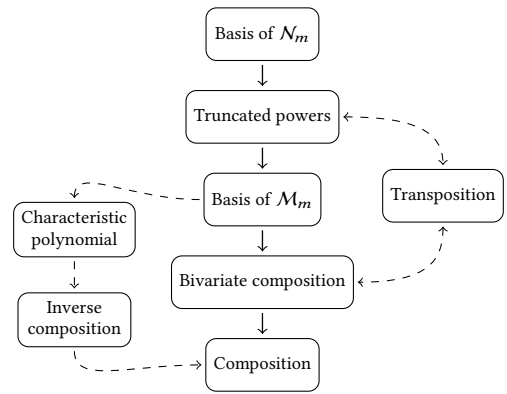

\subsection{Truncated powers vs bivariate composition}
\label{sec:truncated_vs_bivariate}
The transposition principle for linear straight-line programs~\cite
[Thm\,13.20]{BCS97} states that the product of a matrix by a vector
has essentially the same cost as the product of the transpose of this
matrix by a vector. This leads to complexity equivalences
between linear problems that are solved by straight-line programs,
like our
algorithms in \cref{sec:bivmodcomp,sec:trunc_powers}. The equivalence
in this section is proved by looking more closely at the linear maps
involved in the computations.

For two polynomials $t,u \in \field[x]$ with \(\deg(u) = l\), 
we denote by $M_{t,u} \in \matspace{l}{l}$ the matrix
of multiplication by $t$ modulo~$u$ in the basis $(1, x, \ldots, x^
{l-1})$.
The problems of bivariate modular composition and truncated powers
(with the parameters $m$ and $d$)
may be addressed using the following block Krylov matrices~\cite[Sec.\,3.4]{NSSV24}:
\[
K_{m,d} =
\begin{bmatrix}
  X & M_{a,f}X & \cdots & M_{a,f}^{d-1}X
\end{bmatrix}
  \in \field^{n \times (md)}
\]
and
\[
  L_{m,d}=\begin{pmatrix} \trsp{X}\\ \vdots\\ \trsp{X}M_{a,f}^{d-1}\\ \end{pmatrix}
  \in\field ^{(md)\times n},
\]
where $X = \trsp{[\ident{m}  ~0]} \in \field^{n\times m}$.
For a given $a$, the matrix $K_{m,d}$ represents the linear map of bivariate composition
for \(G\in \field[x,y]_{<(m,d)}\)
as in \cref{prop:bivmodcomp}, while $L_{m,d}$ represents that of truncated powers
for \(b\in \field[x]_{<n}\), as in \cref{prop:truncatedpowers}. Note
that $L_{m,d}\neq\trsp{K_{m,d}}$. Still, we now argue that truncated
powers and bivariate composition are essentially equivalent problems when $f(0)\neq 0$.

The following lemma makes use of two auxiliary matrices. The 
\emph{reversal matrix} $J_m\in \field^{m\times m}$ is the matrix of the permutation $
(m,\dots,1)$. The triangular Hankel matrix~$S\in \field^
{n\times n}$ is 
defined by $S_{i,j}=f_{i+j-1}$ if $(n-i+1) \geq  j$
and $S_{i,j} = 0$ otherwise. Note that $S$ is invertible, since $\deg
(f) = n$. It is a symmetrizer for $a$
\cite[p.\,455]{LT85}, that is, it satisfies~\cite[Lem.\,2.5]{BJMS17}
\begin{equation} \label{eq:symmetrizer}
 S \trsp{M}_{a,f} = M_{a,f} S. 
\end{equation}
\begin{lemma}\label{lem:KTL}
For $m\in \{1, \ldots, n\}$ and $d\geq 1$ we have
\begin{equation}\label{eq:blockKrylov_T}
L_{m,d} \cdot  P_n=\diag{Q_m, \ldots, Q_m} \cdot \trsp{K_{m,d}} 
\end{equation}
with 
 $Q_m=-M_{f,x^m} J_m $ and
$P_n= M_{x^m,f} S \in \field^{n\times n}$; $Q_m$ and $P_n$  depend only on~$f$ 
and are invertible if $f(0)\neq 0$.
\end{lemma}
\begin{proof}
The matrix $Q_m$ arises from the product
\cite[Eq.\,(5), p.\,456]{LT85}
\[
M_{x^m,f} S = \begin{pmatrix} Q_m& 0\\0 & T_{n-m}\end{pmatrix} \in \field^{n\times n},
\]
with $T_{n-m}
=J_{n-m} M_{\hat f,x^{n-m}}$ with $\hat f=x^n f(1/x)$.
Hence
\[
Q_m \, \trsp{X} \, S^{-1} = \trsp{X}  \, M_{x^m,f}.
\]
Using \cref{eq:symmetrizer} with~\(a^k\), for any integer \(k \ge 0\), gives
\begin{equation} \label{eq:XTH}
\trsp{X} \trsp{M}_{a^k,f} = \trsp{X} S^{-1} M_{a^k,f} S.
\end{equation}
Combining both equations gives
$
 Q_m \,\trsp{X} \trsp{{M}_{a^k,f}} = \trsp{X}  \, M_{x^m,f} {{M}_{a^k,f}} S,
$
for an arbitrary $k \geq 1$,
which is also
\[
Q_m \, (\trsp{X} \trsp{{M}_{a^k,f}}) =
 (\trsp{X}  \,{{M}_{a^k,f}})  \,P_n  \in \field^{m\times n},
\]
where $P_n= M_{x^m,f} S \in \field^{n\times n}$. This concludes the
proof of \cref{eq:blockKrylov_T}.

Up to sign, the determinants of $Q_m$ and $T_{n-m}$ are $f_0^m$ and
$f_n^{n-m}$. Thus both $Q_m$ and $P_n$ are invertible if $f(0)\neq 0$.
\end{proof}

Applying the transposition principle to the algorithm in 
\cref{sec:bivmodcomp} gives a multiplication 
by $\trsp{K_{m,d}}$ in essentially the same complexity as bivariate
modular composition. Multiplying $\diag{Q_m, \ldots, Q_m}$,
$P_n$, and their inverses by vectors has quasilinear complexity by
reduction to polynomial multiplication
\cite[Ch.\,2]{BiPa94}, hence \cref{eq:blockKrylov_T} gives an algorithm of the same
complexity
for truncated powers.
The reverse direction works in the same way, by transposing the 
algorithm in \cref{sec:trunc_powers} in order to multiply $\trsp{L_{m,d}}$ by vectors. 

There exist techniques to transpose algorithms in a systematic
way~\cite{BLS03}. We have not checked whether they would produce
exactly the same algorithms as those presented in \cref{sec:bivmodcomp,sec:trunc_powers}.

\subsection{Characteristic polynomial vs composition} \label{sec:charpoly_vs_composition}

For any $f$ and a generic $a$, the determinant of a minimal basis 
of~$\module_\mm$, made monic, is the characteristic polymonial $\chi _a$
of $a$ modulo~$f$ \cite[Prop.\,10.1]{NSSV24}.
\Cref{prop:genericbasisMm} can thus be extended to the computation of
$\chi _a$, using also 
$\softO{\mm^{\expmm-1}n}$ operations \cite{LVZ17}. 
In \cref{appendix:charpoly_vs_composition}, we show that the
characteristic polynomial and the modular composition problems are
essentially equivalent in the model of computation by
straight-line programs.
In particular, a technique commonly used for parameterizing algebraic
varieties enables us to reduce the modular composition problem
and its inverse (see
[\citealp[Thm.\,3.5]{Shoup94};
\citealp[Sec.\,6.2]{NSSV24}]) to characteristic polynomial
computations. 

%%%%%%%%%%%%%%%%%%%%%%%%%%%%%%%%%%%%%%%%%%%%%%%%%%%%%%%%%%%%%%%%%%%%%%%%%%%%%%%%%%%
%
%              BIBLIO
%
%%%%%%%%%%%%%%%%%%%%%%%%%%%%%%%%%%%%%%%%%%%%%%%%%%%%%%%%%%%%%%%%%%%%%%%%%%%%%%%%%%%


\newcommand{\Hoeven}{\relax}\newcommand{\Gathen}{\relax}

%%%%%%%%%%%%%%%%%%%%%%%%%%%%%%%%%%%%%%%%%%%%%%%%%%%%%%%%%%%%%%%%%%%%%%%%%%%%%%%%%%%
%
%              APPENDIX
%
%%%%%%%%%%%%%%%%%%%%%%%%%%%%%%%%%%%%%%%%%%%%%%%%%%%%%%%%%%%%%%%%%%%%%%%%%%%%%%%%%%%

\appendix

\section{Generic bases of \texorpdfstring{$\nodule _\mm$}{Nmu}: proof of Lemma \ref{lem:generic_bases_Nm}}
\label{appendix:genericity}

The proof is based on fairly standard results, such as the
links between Popov forms and controllability realizations \cite[Scheme\,II,
p.~427 \& Sec.\,6.7.2, p.~481]{Kailath80}.
Let $f \in \xring$ be of degree~$n$. By identifying elements of the
quotient
$a\in\xring/\genBy{f}$ with the corresponding coefficient
vector~$v_a\in\field^n$, we view $\field^n$ as a $\xring$-module through
$x\cdot v_a = v_{xa}$. We denote by $M_x \in \matspace{n}{n}$ the matrix of
multiplication by~$x$ modulo $f$ in the basis $(1,x,\ldots, x^{n-1})$, so that $x \cdot v_{a} = M_x v_a$. This
allows us to relate polynomial relations between the powers of \(a \bmod f\) to
linear dependencies between vectors in Krylov
subspaces~\cite[Sec.\,3.10]{Jac85}.

\begin{lemma}
  \label{lem:mindeg_Nr}%
  Let $V_a$ be the matrix $[v_1 ~v_a ~\cdots ~v_{a^{\mm-1}}] \in
  \field^ {n \times \mm}$. The degree of a minimal basis of $\nodule_\mm$ is the
  smallest integer $\dd$ such that
  \begin{align*}
    \operatorname{rank}([V_a ~ & M_x V_a ~\cdots ~M_x^{\dd-1} V_a]) \\
                               & = \dim (\Span(\{M_x^k v_{a^j}, k \in \mathbb N, 0 \le j < \mm\}))
                                 = n.
  \end{align*}
\end{lemma}
\begin{proof}
  The latter dimension is indeed \(n\), since the span contains 
  the independent vectors
  \(\{M_{x}^k v_1, 0 \le k < n\} = \{v_1,v_x,\ldots,v_{x^{n-1}}\}\).

  We first show that the degree of a minimal basis of~$\nodule_\mm$ cannot be less than $\dd$.
  By contradiction, assume there exists a minimal basis $R \in \xmatspace{\mm}{\mm}$
  of $\nodule_\mm$ whose degree is \(\ell < \dd\). Let \((\ell_1,\ldots,\ell_\mm)\) be
  the column degrees of \(B\) and \(L \in \matspace{\mm}{\mm}\) be its leading
  matrix, that is, \(L_{i,j}\) is the coefficient of degree \(\ell_j\) of
  \(R_{ij}\). Since \(R\) is column reduced, \(L\) is invertible. Then
  \(
  P = R \diag{x^{\ell-\ell_1}, \ldots, x^{\ell - \ell_\mm}} L^{-1}
  \)
  has the form
  \[
    P = P_0 + xP_1 + \cdots + x^{\ell-1} P_{\ell-1} + x^\ell I_\mm,
    \quad P_k \in \matspace{\mm}{\mm}.
  \]
  Since the columns of \(R\) are in \(\nodule_\mm\), those of \(P\) are too,
  meaning that 
  \begin{equation} \label{eq:main_appA}
  [1 \;\; a \;\; \cdots \;\; a^{m-1}] R = 
  [1 \;\; a \;\; \cdots \;\; a^{m-1}] P =0 \bmod f.
  \end{equation}
   This translates as
  $
    V_a P_0 + M_x V_a P_1 + \cdots + M_x^{\ell-1} V_a P_{\ell-1} + M_x^{\ell} V_a = 0.
  $
  From this identity, any vector in \(\{M_x^k v_{a^j}, k \ge \ell, 0 \le
  j < \mm\}\) can be expressed as a linear combination of the columns
  of the matrix \([V_a ~ M_x V_a ~\cdots ~M_x^{\ell-1} V_a]\),
  which contradicts the minimality of \(\dd\).

We follow a similar path to show that a minimal basis of \(\nodule_\mm\) has
degree at most \(\dd\). By definition of \(\dd\), the columns of $M_x^\dd V_a$ are
combinations of those of $[V_a ~ M_x V_a ~\cdots ~M_x^{\dd-1} V_a]$. Therefore
there exist matrices \(P_0,\ldots,P_{\dd-1} \in \matspace{\mm}{\mm}\) such that
$
  V_a P_0 + M_x V_a P_1 + \cdots + M_x^{\dd-1} V_a P_{\dd-1} + M_x^{\dd} V_a = 0.
$
By construction, the matrix \(P = P_0 + \cdots + x^{\dd-1} P_{\dd-1} + x^\dd I_\mm\) is
nonsingular and its columns are in \(\nodule_\mm\). It follows, by minimality,
that any minimal basis of \(\nodule_\mm\) has degree at most \(\deg(P) = \dd\).
\end{proof}

\cref{lem:generic_bases_Nm} can now be proven by considering~$a=x^{\dd}$, where \(\dd = \lceil n/\mm\rceil\). Decompose $f$ as $f=\sum_{0
\le i < \mm} f^{(i)} a^{i}$, with $\deg(f^{(i)}) < \dd$ for \(i < \mm-1\)  and
\(\deg(f^{(\mm-1)}) \le \dd\). In the identity
\begin{equation*}
  \left[\begin{smallmatrix}
    f & -a & -a^2 & \cdots & -a^{\mm-1} \\
      & 1  \\
      &    & 1 \\
      &    &      & \sddots \\
      &    &      &        & 1
    \end{smallmatrix}\right]
    \left[\begin{smallmatrix}
    0 & 0 & \cdots & 0 & 1 \\
    -1  & a  &&& f^{(1)}\\
        &  -1  & \sddots && \svdots\\
        &    &    \sddots  & a& \svdots \\
        &    &      &   -1     & f^{(\mm-1)}
  \end{smallmatrix}\right]=
 \left[\begin{smallmatrix}
    a & 0 & \cdots & 0 & f^{(0)} \\
    -1  & a  &&& f^{(1)}\\
        &  -1  & \sddots && \svdots\\
        &    &    \sddots  & a& \svdots \\
        &    &      &   -1     & f^{(\mm-1)}
  \end{smallmatrix}\right]
\end{equation*}
in $\xmatspace{\mm}{\mm}$, the first matrix in the product is a basis
 of~$\nodule_\mm$ and the second one is
unimodular. It follows that the right-hand side is a basis of $\nodule _\mm$ of
degree~$\dd$, and that minimal bases of $\nodule_\mm$ have degree at most~$\delta$. 
Hence, by \cref{lem:mindeg_Nr}, $[V_a ~ M_x V_a ~\cdots ~M_x^{\dd-1}
V_a]$ has maximal rank~$n$. Now consider a generic \(a\). A fortiori, the
latter matrix still has rank~$n$, and by \cref{lem:mindeg_Nr} again, any minimal
basis of $\nodule_\mm$ has degree at most $\dd$. 
Finally, since the sum of the column degrees of a minimal basis of~$\nodule_\mm$ is exactly the degree of its determinant, i.e. $\deg(f) = n$,  such a basis cannot have a degree less than $\lceil n/\mm\rceil$. 
Therefore, minimal bases of $\nodule_m$ have degree exactly~$\lceil
n/\mm\rceil$.
Based on the above, if $\bar a = \sum_{0 \le i < n} \bar a_i x^i$, 
where the $\bar a_i$’s are new
indeterminates, then  $[V_{\bar a} ~ M_x V_{\bar a} ~\cdots ~M_x^{\dd-1}
V_{\bar a}]$ has a nonzero $n$-minor $\witnessN_{f,\mm} \in \field[\bar a_0,\ldots \bar a_{n-1}]$. 
A set of polynomials $a$ for which minimal bases of $\nodule_m$ have degree exactly~$\dd=\lceil
n/\mm\rceil$ is thus given by those whose coefficients avoid the zero set of~$\witnessN_{f,\mm}$.

%%%%%%%%%%%%%%%%%%%%%%%%%%%%%%%%%%%%%%%%%%%%%%%%%%%%%%%%%%%%
%%%%%%%%%%%%%%%%%%%%%%%%%%%%%%%%%%%%%%%%%%%%%%%%%%%%%%%%%%%%
%%%%%%%%%%%%%%%%%%%%%%%%%%%%%%%%%%%%%%%%%%%%%%%%%%%%%%%%%%%%

\section{Generic bases of \texorpdfstring{$\module_\mm$}{Mmu}}
\label{appendix:genericityMm}

We show that for any $f\in \field[x]$ of degree $n$, and for generic $a$, any minimal basis 
of $\module_\mm$ defined by \cref{eq:Mm} has degree $\dd=\lceil n/\mm\rceil$.
We use techniques very similar to those in  \cref{appendix:genericity}, which we do not repeat in detail.
Let $M_a$ be the matrix of multiplication by~$a$ modulo~$f$.
We now view $\field ^n$ as a $\field[y]$-module through $y \cdot v= M_av$ for $v\in \field ^n$. 
By swapping the roles of $a$ and $x$ in \cref{lem:mindeg_Nr} we obtain the following. 
Since the minimal polynomial of $a$  may have a degree less than $n$, 
the statement now involves a dimension that may not be equal to $n$. 
However, this dimension will be $n$ in generic cases. 

\begin{lemma}
  \label{lem:mindeg_Mr}%
  Let $V_x = [v_1 ~v_x ~\cdots ~v_{x^{\mm-1}}] = \trsp{[\ident{\mm}  ~0]} \in
  \field^{n \times \mm}$. The degree of a minimal basis of $\module_\mm$ is the
  smallest integer $\dd$ such that
  \begin{align*}
    \operatorname{rank}([V_x ~ & M_a V_x ~\cdots ~M_a^{\dd-1} V_x]) \\
                               & = \dim (\Span(\{M_a^k v_{x^j}, k \in \mathbb N, 0 \le j < \mm\})) \leq n.
  \end{align*}
\end{lemma}

\cref{lem:mindeg_Mr} can be proved by reasoning analogous to that of the proof of \cref{lem:mindeg_Nr}. 
Relation matrices $P$ for $\module _\mm$ are now matrices in~$\field[y]^{\mm \times \mm}$.  
\Cref{eq:main_appA}
is reinterpreted in the form
$$
[1 \;\; x \;\; \cdots \;\; x^{\mm-1}] P =0 \bmod \idealGens, 
$$
which now translates as 
$
    V_x P_0 + M_a V_x P_1 + \cdots + M_a^{\ell-1} V_x P_{\ell-1} + M_a^{\ell} V_x P_{\ell}= 0,
$
  when $P$ has degree $\ell$.
The degree of minimal bases of $\module_\mm$ in generic cases follows from \cref{lem:mindeg_Mr} by considering $a$ of degree $\mm$.
For such an $a$, the matrix $[V_x ~  M_a V_x ~\cdots ~M_a^{\dd-1} V_x]$ has 
rank~$n$. 
Indeed, since $a$ has degree $\mm$, its first $n$ columns form an invertible upper triangular matrix.
It can finally be deduced that there exists a polynomial 
 $\Psi_{f,\mm} \in \field[\bar a_0,\ldots \bar a_{n-1}]$ whose avoidance of the zero set makes it possible to 
characterize polynomials $a$ for which minimal bases of~$\module_m$ have degree exactly~$\dd=\lceil
n/\mm\rceil$.

\section{Characteristic polynomial vs composition} \label{appendix:charpoly_vs_composition}

\balance

Here, we argue informally that for generic $a$ and $f$ (in the Zariski sense), the problem of computing the characteristic polynomial of $a$ modulo $f$ and modular composition have roughly the same complexity.
We restrict to algebraic algorithms that can be written as straight-line programs.
In this case, each problem can be solved by an algorithm whose cost is
a constant multiple of that of the other problem,
plus $\softO{n}$ operations.
We then explain how our approach could exploit this equivalence.
Note that, generically,~$f$ and the characteristic polynomial $\chi _a$ of $a$ are separable, hence
the minimal polynomial of $a$ modulo $f$ coincides with~$\chi _a$.

The reduction of the  (characteristic) minimal  polynomial  to 
composition is frequently used; it can be inferred from Shoup's work~\cite{Shoup94,Shoup99}. The minimal polynomial of $a$ modulo $f$ can be computed using a randomized algorithm,
which requires one modular composition (allowing $g$ to have degree $2n-1$) and the computation of the minimal polynomial of a linearly recurrent sequence in~$\field^{\mathbb N}$~\cite[Sec.\,12.3]{GaGe97}.
When $\chi _a$ is separable, a deterministic algorithm can be derived using the trace instead of random linear forms as was done in \cite{Shoup99}. (See, for example, \cite[Sec.\,4.3]{PS13}).

We now sketch the reduction in the opposite direction, via the problem of  inverse modular composition.
Given $a$ and $f$ as before, and for a polynomial $h\in \field[x]_{<n}$, inverse modular composition is the problem of computing $g \in \field[y]_{<n}$ such that $g(a)=h \bmod f$ [\citealp[Thm.\,3.5]{Shoup94};
\citealp[Sec.\,6.2]{NSSV24}]. Assuming that the minimal polynomial of
$a\bmod f$ has degree $n$, the polynomial $g$ always exists and is
unique.
To address inverse modular composition, we introduce a new variable~$z$ and consider the polynomial $a+zh$.
Using that~$\chi _a \in \field[y]$ is separable,  $h$ and the characteristic polynomial $\chi_{a,h}\in \field[y,z]$ of $a+zh$ modulo~$f$ can be related as follows (\cite[Lem.\,2.8]{ABRW96} or \cite[Sec.\,12.4, Eq. 12.4]{BPR06}, 
valid for an arbitrary field):
\begin{equation} \label{eq:uChow}
-\frac{\left((\frac{\partial}{\partial z} \chi_{a,h})_{z=0}\right)(a)}{\chi' _a(a)} = h \bmod f.
\end{equation}
If $\bar\chi _a$ is the inverse of $\chi'_a$ modulo $\chi_a$, then
\begin{equation} \label{eq:uChow2}
g= -\bar\chi _a \left(\frac{\partial}{\partial z} \chi_{a,h}\right)_{z=0} \rem \chi_a
\end{equation}
is the solution to the problem of inverse composition for $h$.

A straight-line program ${\mathcal C}$ that computes $\chi_a$
gives rise to a program ${\mathcal C}_{z}$ that computes $\chi_{a,h}$
modulo $\langle f, z^2\rangle$, giving
the derivative at $z=0$ in the left-hand side of 
\cref{eq:uChow}. The divisions that occur in~$
{\mathcal C}_ {z}$ are those already present in ${\mathcal C}$
and the overall cost is multiplied by a constant (see e.g., the proof of \cite[Lem.\,(7.2)]{BCS97}).
Note that here we have considered straight-line programs, but a situation where the tests on elements modulo $z^2$
 can be done on the coefficients of valuation $0$ could be suitable.
Once the above derivative is known, the solution~$g$ from 
\cref{eq:uChow2} can be computed after applying the extended Euclidean
algorithm to produce~$\bar \chi _a$.

Finally, modular composition can be reduced to inverse modular composition. Indeed,
if $\alpha \in \field[y]_{<n}$ and $\gamma \in \field[x]_{<n}$ satisfy
\begin{equation}\label{eq:comp1}
\alpha(a)=x \bmod f, ~\gamma(\alpha)=g \bmod \chi_a,
\end{equation}
then we obtain $\gamma = g(a) \rem f$ in two inverse modular
compositions within the announced complexity.

We remark that \cref{eq:comp1} defines the $\field$-algebra isomorphism
\begin{equation} \label{eq:defKiso}
\phi_a: \field[x]/\langle f\rangle \rightarrow \field[y]/\langle \chi _a\rangle
\end{equation}
that maps $x$ to $\alpha$ and $a$ to $y$, and more generally $u$ to $v$ such that $v(a)=u \bmod f$.
Therefore, $p(\alpha,y) = 0 \bmod \chi_a$, with $p \in \xyring$, if and only if
$p(x,a) = 0 \bmod f$. This implies, in particular,  that
the bases of~$\module_\mm$ and $\nodule_\mm$ for $\alpha$ modulo $\chi _a$ (after the
substitution of “$y$” by “$x$”, and “$x$” by “$y$” to make notation match), coincide with the bases
of~$\nodule_\mm$ and $\module_\mm$ for $a$ modulo $f$, respectively.

We now discuss our composition approach in light of this
computational equivalence.
The determinant of a minimal basis
of~$\module_\mm$ is a multiple of the minimal polynomial of $a\bmod f$
\cite[Prop.\,4.1]{NSSV24}, 
hence gives the characteristic polynomial when the two coincide.  
Since this determinant can be computed using $\softO{\mm^{\expmm-1}n}$ operations \cite{LVZ17}, obtaining $\chi_a$ does not
incur an additional cost beyond the minimal basis cost.
According to \cref{eq:comp1}, modular composition reduces to two applications of \cref{eq:uChow2}, which
essentially corresponds to the computation of the characteristic polynomials of $a+zx$ modulo $\langle f, z^2\rangle$
and $\alpha+zg$ modulo $\langle \chi _a, z^2\rangle$.
To the extent that the computation of a minimal basis of $\module _m$ and its determinant are written in terms
of straight-line programs (the isomorphism of \cref{eq:defKiso} ensures the relevance for both $a$ and $\alpha$), we obtain
a different version of the last part of our algorithm, where 
computing the two characteristic polynomials above  and applying
\cref{eq:uChow2,eq:comp1}
yields the result of the
final bivariate modular composition of Steps
\ref{new-algo:step4}-\ref{new-algo:last-composition}.

\end{document}